\documentclass[onecolumn,12pt]{IEEEtran}
\usepackage{epsfig,latexsym}
\usepackage{float}
\usepackage{indentfirst}
\usepackage{amsmath}
\usepackage{amssymb}
\usepackage{times}
\usepackage{subfigure}
\usepackage{psfrag}
\usepackage{cite}
\usepackage{lastpage}
\usepackage{algorithm}
\usepackage{algorithmic}
\usepackage{hyperref}
\usepackage{fancyhdr}
\usepackage{color}
 \usepackage{amsthm}
\usepackage{bigints}
\def\diag{\textrm{diag}}

\newtheorem{Lemma}{Lemma}
\newtheorem{Corollary}{Corollary}
\newtheorem{lemma}[Lemma]{$\mathbf{Lemma}$}

\newtheorem{corollary}[Corollary]{$\mathbf{Corollary}$}
\begin{document}
\title{ {\LARGE   A General MIMO Framework  for NOMA Downlink and Uplink Transmission Based on Signal Alignment}}

\author{ Zhiguo Ding, \IEEEmembership{Member, IEEE},     Robert Schober, \IEEEmembership{Fellow, IEEE},
  and  H. Vincent Poor, \IEEEmembership{Fellow, IEEE}\thanks{
Z. Ding and H. V. Poor  are with the Department of
Electrical Engineering, Princeton University, Princeton, NJ 08544,
USA.   Z. Ding is also with the School of
Computing and Communications, Lancaster
University, LA1 4WA, UK.  R. Schober is with the Institute for Digital
Communications, University of Erlangen-Nurnberg, Germany. }\vspace{-4em}} \maketitle
\begin{abstract}
    The application of multiple-input multiple-output (MIMO) techniques to non-orthogonal multiple access (NOMA) systems is important to   enhance the performance gains of NOMA. In this paper, a novel MIMO-NOMA framework for downlink and uplink transmission is proposed by applying the concept of signal alignment. By using stochastic geometry, closed-form analytical results are developed to facilitate the performance evaluation of  the proposed framework for randomly deployed users and interferers. The impact of different power allocation strategies, such as fixed power allocation and cognitive radio inspired power allocation, on the performance of MIMO-NOMA is also investigated. Computer simulation results are provided to demonstrate the performance of the proposed framework and the accuracy of the developed analytical results.
\end{abstract}\vspace{-1em}
\section{Introduction}
Non-orthogonal multiple access (NOMA) has been recognized as a spectrally efficient multiple access (MA) technique for the next generation of mobile networks \cite{6692652, NOMAPIMRC,Nomading}. For example,  the use of NOMA has   been recently proposed for    downlink scenarios in 3rd generation partnership project  long-term evolution (3GPP-LTE) systems, and the considering technique was termed  multiuser superposition transmission (MUST)   \cite{3gpp1}. In addition, NOMA has also  been identified  as one of the key radio access  technologies  to increase system capacity and reduce latency in fifth generation (5G) mobile networks   \cite{docom}, \cite{metis}.

The key idea of NOMA is to exploit the power domain for multiple access, which means multiple users can be served concurrently at the same time, frequency, and spreading code. Instead of using water-filling power allocation strategies, NOMA   allocates more power to the  users with poorer channel conditions, with the aim to facilitate  a balanced tradeoff between system throughput and user fairness. Initial system implementations of NOMA in cellular networks have  demonstrated the superior spectral efficiency of NOMA  \cite{6692652}, \cite{NOMAPIMRC}. The performance of NOMA in a network with randomly deployed  single-antenna nodes was investigated in \cite{Nomading}. User fairness in the context of NOMA has been addressed in \cite{Krikidisnoma}, where  power allocation was optimized under different channel state information (CSI) assumptions. In \cite{timnoma}, topological interference management has been applied for single-antenna downlink NOMA transmission. Unlike the above works,  \cite{6933459} addressed the application of NOMA for uplink transmission, where   the problems of power allocation and subcarrier allocation were jointly optimized. The concept of NOMA is not   limited to radio frequency communication networks, and has been recently applied to visible light communication systems in \cite{vlcnoma}.

The application of multiple-input multiple-output (MIMO) technologies to NOMA is important since the use of MIMO provides additional  degrees of freedom for further performance improvement. In \cite{7015589}, the multiple-input single-output  scenario, where  the base station had multiple antennas and users were  equipped with a single antenna, was considered.  In \cite{7095538}, a multiple-antenna  base station used  the NOMA approach to serve two multiple-antenna users simultaneously, where the problem of   throughput maximization was formulated and two  algorithms were proposed to solve the  optimization problem. In many practical scenarios, it is preferable to serve as many users as possible in order to reduce user latency and improve user fairness. Following this rationale, in \cite{Zhiguo_mimoconoma}, users were first grouped into small-size clusters, where NOMA was implemented for  the users within one cluster and MIMO detection was used to cancel inter-cluster interference. Similar  to \cite{6692307}, this method does not need  CSI at the base station; however, unlike \cite{6692307}, it avoids the use of random beamforming which can cause  uncertainties for the quality of service (QoS) experienced by the  users.

This paper   considers a general MIMO-NOMA communication network where  a base station is communicating with multiple users using the same time, frequency, and spreading code resources, in the presence of randomly deployed interferers. The contributions of this paper are listed as follows:
\begin{itemize}
\item A general MIMO-NOMA framework which is applicable to both downlink and uplink transmission is proposed, by applying the concept of signal alignment, originally  developed for multi-way relaying channels in \cite{Lee100} and \cite{6384814}. By exploiting  this framework, the considered  multi-user MIMO-NOMA scenario can be decomposed into multiple separate single-antenna NOMA channels, to which conventional NOMA protocols can be applied straightforwardly.

\item Since  the choice of the power allocation coefficients  is key to achieve a favorable throughput-fairness tradeoff in NOMA systems, two types of power allocation strategies are studied in this paper. The fixed power allocation strategy can realize different QoS requirements in the long term, whereas the cognitive radio inspired power allocation strategy can ensure that users' QoS requirements are met instantaneously.

\item A sophisticated approach for the user precoding/detection vector selection is proposed and combined with the signal alignment framework in order to efficiently exploit the excess degrees of freedom of the MIMO system.  Compared to the existing MIMO-NOMA work in \cite{Zhiguo_mimoconoma}, the  framework proposed in this paper offers two benefits. First,  a larger diversity gain can be achieved, e.g., for a scenario in which all nodes are equipped with $M$ antennas, a diversity order of $M$ is achievable, whereas a diversity gain of $1$ is realized by the scheme in   \cite{Zhiguo_mimoconoma}. Second,  the proposed framework is more general, and   also applicable to the case where the users have fewer  antennas than the base station.

\item Exact expressions and asymptotic performance results are developed in order to obtain an insightful understanding of the proposed MIMO-NOMA framework. In particular, the outage probability is used as the performance criterion since it not only bounds the error probability of detection tightly, but also   can be used to calculate the outage capacity/rate. The impact of the random locations of the  users and the interferers is captured by applying   stochastic geometry, and the diversity order is computed  to illustrate how efficiently the degrees of freedom of the channels are used by the proposed framework.
\end{itemize}

\section{System Model for the Proposed  MIMO-NOMA Framework}\label{section system model}
Consider an  MIMO-NOMA downlink (uplink) communication scenario in which a base station is communicating with multiple  users. The base station is equipped with $M$ antennas and each user is equipped with $N$ antennas. In this paper, we consider the scenario $N> \frac{M}{2}$ in order to implement the concept of signal alignment, an assumption more general than the one used in \cite{Zhiguo_mimoconoma}.  This assumption is applicable to various communication scenarios, such as small cells in heterogenous networks \cite{7070674} and 5G cloud radio access networks \cite{CMCC}, in which low-cost  base  stations are deployed with high density  and it is reasonable to assume that the base stations have capabilities similar to those of user handsets, such as smart phones and tablets.

The users are assumed to be uniformly deployed in a disc, denoted by $\mathcal{D}$, i.e., the cell controlled by the base station. The radius of the disc is $r$, and the base station is located at the center of $\mathcal{D}$. In order to reduce the system load, many existing studies about NOMA have proposed to pair two users for the implementation of NOMA, and have demonstrated that it is ideal to pair two users whose channel conditions are very different  \cite{6692652}, \cite{Zhiguo_CRconoma}. Based on this insight, we assume  that the disc is divided into two regions. The first region  is a smaller disc, denoted by $\mathcal{D}_1$, with   radius   $r_1$ ($r_1<r$) and the base station located at its origin. The second region  is a ring, denoted by $\mathcal{D}_2$, constructed from $\mathcal{D}$ by removing $\mathcal{D}_1$. Assume  that $M$ pairs of users are selected, where  user $m$, randomly located in $\mathcal{D}_1$, is paired with  user $m'$, randomly located in $\mathcal{D}_2$. Hence,  the users are randomly scheduled and paired together. The use of more sophisticated schedulers can further improve the performance of the proposed MIMO-NOMA framework of course, but this is beyond the scope of this paper.

In addition to the  messages sent by the base station, the downlink NOMA users also observe signals sent by   interference sources which are distributed in $\mathcal{R}^2$ according to a homogeneous Poisson point process (PPP)   $\Psi_I$ of density $\lambda_{I}$ \cite{Haenggi}. The same assumption is made for the uplink case. In practice, these interferers   can be cognitive radio transmitters, WiFi access points in LTE in the unlicensed spectrum (LTE-U), or transmitters from different tiers in heterogenous networks.   In order to obtain tractable analytical results, it is assumed that the interference sources are equipped with a single antenna and use identical  transmission powers, denoted by $\rho_{I}$.

Consider the use of a composite channel model with both  quasi-static Rayleigh fading and large scale path loss. In particular, the channel matrix from the base station to  user $m$ is $\mathbf{H}_m=\frac{\mathbf{G}_m}{\sqrt{L(d_m)}}$, where $\mathbf{G}_m$ denotes an $N\times M$ matrix whose elements represent Rayleigh fading channel gains, $d_m$ denotes the distance from the base station to the user, and the resulting  path  loss is modelled as follows:
\begin{eqnarray}
L(d_m)=\left\{\begin{array}{ll}d_m^\alpha, &\text{if} \quad d_m>r_0 \\ \nonumber r_0^\alpha,& \text{otherwise}\end{array}\right.,
\end{eqnarray} where  $\alpha$ denotes the path loss exponent and     parameter $r_0$  avoids a singularity when the distance is small. It is assumed that $r_1\geq r_0$ in order to simplify the analytical results.  For notational simplicity,  the channel matrix from user $m$ to the base station is denoted by $\mathbf{H}_m^H$. Global CSI is assumed to be available  at the users and the base station.  The proposed MIMO-NOMA framework for downlink and uplink transmission is described in the following two subsections, respectively.

\subsection{Downlink MIMO-NOMA Transmission}
The base station  sends the following $M \times 1$ information-bearing vector
 \begin{align}
 \mathbf{s}=\begin{bmatrix}\alpha_{1}s_{1} + \alpha_{1'}s_{1'} \\ \vdots \\ \alpha_{M}s_{M} + \alpha_{M'}s_{M'}  \end{bmatrix},
 \end{align}
 where $s_m$ is the signal intended for the $m$-th user,  $\alpha_m$ is the power allocation coefficient, and $\alpha_{m}^2+\alpha_{m'}^2=1$. The choice of the power allocation coefficients will be discussed later.

 Without loss of generality, we focus on  user $m$, whose observation is give by
 \begin{align}
 \mathbf{y}_m = \frac{\mathbf{G}_m}{\sqrt{L(d_m)}} \mathbf{P} \mathbf{s} +\mathbf{w}_{I_m}+\mathbf{n}_m,
 \end{align}
 where $\mathbf{P}$ is the $M\times M$ precoding matrix to be defined at the end of this subsection, $\mathbf{w}_{I_m}$ denotes the overall co-channel interference received by  user $m$, and $\mathbf{n}_m$ denotes the noise vector. Following the classical shot noise model in \cite{4086349}, the co-channel interference, $\mathbf{w}_{I_m}$, can be expressed as follows:
 \begin{align}
 \mathbf{w}_{I_m}\triangleq \underset{j\in\Psi_I}{\sum}\frac{\sqrt{\rho_I}}{\sqrt{L(d_{I_j,m})}}\mathbf{1}_{N},
  \end{align}
  where $\mathbf{1}_m$ denotes an $m\times 1$ all-one vector, and $d_{I_j,m}$ denotes the distance from  user $m$ to the $j$-th interference source. Note that small scale fading has been omitted in the interference model, since the effect of path loss is more dominant for interferers located far away. In addition, this simplification will facilitate the development of  tractable analytical results.   The case with $\rho_{I}=0$ corresponds to the scenario  without interference.

   User $m$   applies a detection vector $\mathbf{v}_m$ to its observation, and therefore the user's observation can  be re-written as follows:
  \begin{align}
 \mathbf{v}_m^H\mathbf{y}_m &= \mathbf{v}_m^H \frac{\mathbf{G}_m}{\sqrt{L(d_m)}}  \mathbf{P} \mathbf{s} +\mathbf{v}_m^H(\mathbf{w}_{I_m}+\mathbf{n}_m)\\ \nonumber &= \mathbf{v}_m^H \frac{\mathbf{G}_m}{\sqrt{L(d_m)}}  \mathbf{p}_m (\alpha_{m}s_{m} + \alpha_{m'}s_{m'})  +\underset{{\rm interference~ (including ~inter-pair ~interference)  ~+~ noise}}{\underbrace{\sum_{i\neq m}\mathbf{v}_m^H \frac{\mathbf{G}_m}{\sqrt{L(d_m)}}  \mathbf{p}_i (\alpha_{i}s_{i} + \alpha_{i'}s_{i'}) +\mathbf{v}_m^H(\mathbf{w}_{I_m}+\mathbf{n}_m)}},
 \end{align}
 where $\mathbf{p}_m$ denotes the $m$-th column of $\mathbf{P}$.

 In order to remove   inter-pair interference,   the following constraint has to be met:
 \begin{align}
 \begin{bmatrix}\mathbf{v}_m^H\mathbf{G}_m \\\mathbf{v}_{m'}^H\mathbf{G}_{m'}\end{bmatrix} \mathbf{p}_i=\mathbf{0}_{2\times 1}, ~ \forall i\neq m,
 \end{align}
where $\mathbf{0}_{m\times n}$ denotes the $m\times n$ all zero matrix. Without loss of generality, we focus on  $\mathbf{p}_1$ which needs to satisfy the following constraint:
   \begin{align}\label{constraint 1}
 \begin{bmatrix}\mathbf{G}_2^H \mathbf{v}_2&\mathbf{G}_{2'}^H\mathbf{v}_{2'}&\cdots&\mathbf{G}_{M}^H\mathbf{v}_{M} &\mathbf{G}_{M'}^H\mathbf{v}_{M'}\end{bmatrix}^H \mathbf{p}_1=\mathbf{0}_{2(M-1)\times 1}.
 \end{align}
Note that the dimension of the matrix in \eqref{constraint 1}, $\begin{bmatrix}\mathbf{G}_2^H \mathbf{v}_2&\mathbf{G}_{2'}^H\mathbf{v}_{2'}&\cdots&\mathbf{G}_{M}^H\mathbf{v}_{M} &\mathbf{G}_{M'}^H\mathbf{v}_{M'}\end{bmatrix}^H$, is $2(M-1)\times M$. Therefore,  a non-zero vector $\mathbf{p}_i$ satisfying \eqref{constraint 1} does not exist. In order to ensure the existence of $\mathbf{p}_i$, one straightforward approach  is to serve less user pairs, i.e., reducing the number of user pairs to $\left(\frac{M}{2}+1\right)$. However, this approach will reduce the overall system throughput.

To overcome this problem, in this paper,  the concept of interference alignment is applied, which means the   detection vectors are designed to  satisfy the following constraint \cite{Lee09},  \cite{Dingtong11}
 \begin{align}\label{constraint 3}
 \mathbf{v}_m^H\mathbf{G}_m =\mathbf{v}_{m'}^H\mathbf{G}_{m'},
 \end{align}
 or equivalently
\begin{align}
\begin{bmatrix} \mathbf{G}_m^H &-\mathbf{G}_{m'}^H\end{bmatrix} \begin{bmatrix} \mathbf{v}_m \\\mathbf{v}_{m'}\end{bmatrix}=\mathbf{0}_{M\times 1}.
 \end{align}
 Define $\mathbf{U}_m$ as the $2N\times (2N-M)$ matrix containing the $(2N-M)$ right singular vectors of  $
\begin{bmatrix} \mathbf{G}_m^H &-\mathbf{G}_{m'}^H\end{bmatrix} $ corresponding to its zero singular values. Therefore, the detection vectors at the users are designed  as follows:
\begin{align}\label{dtection 1}
  \begin{bmatrix} \mathbf{v}_m \\\mathbf{v}_{m'}\end{bmatrix} = \mathbf{U}_m \mathbf{x}_m,
\end{align}
where $\mathbf{x}_m$ is a  $(2N-M)\times 1$  vector to be defined later. We normalize  $\mathbf{x}_m$  to $2$,  i.e., $|\mathbf{x}|^2=2$, due to the following two reasons. First, the uplink transmission power has to be  constrained as shown in the following subsection. Second, this facilitates the performance analysis carried out in the next section. It is straightforward to show that the choice of the detection vectors in \eqref{dtection 1} satisfies $\begin{bmatrix} \mathbf{G}_m^H &-\mathbf{G}_{m'}^H\end{bmatrix} \mathbf{U}_m \mathbf{x}_m=\mathbf{0}_{M\times 1}$.

  The effect of the signal alignment based design in \eqref{constraint 3} is the projection of the channels of the two users in the same pair into the same direction. Define $\mathbf{g}_m\triangleq \mathbf{G}_m^H\mathbf{v}_m$ as the effective channel vector shared by the two users. As a result, the number of the rows in the matrix in \eqref{constraint 1} can be reduced significantly. In particular, the constraint for $\mathbf{p}_i$ in \eqref{constraint 1} can be rewritten as follows:
   \begin{align}\label{constraint 2}
 \begin{bmatrix}\mathbf{g}_1&\cdots &\mathbf{g}_{i-1} &\mathbf{g}_{i+1}&\cdots&\mathbf{g}_M\end{bmatrix}^H \mathbf{p}_i=\mathbf{0}_{(M-1)\times 1}.
 \end{align}
Note that $\begin{bmatrix}\mathbf{g}_1&\cdots &\mathbf{g}_{i-1} &\mathbf{g}_{i+1}&\cdots&\mathbf{g}_M\end{bmatrix}^H $ is an $(M-1)\times M$ matrix, which means that a $\mathbf{p}_i$ satisfying \eqref{constraint 2} exists.

 Define $\mathbf{G}\triangleq \begin{bmatrix}\mathbf{g}_1 &\cdots&\mathbf{g}_M\end{bmatrix}^{H} $.
 A zero forcing based precoding matrix at the base station can be designed as follows:
 \begin{align}\label{p design}
 \mathbf{P } =  \mathbf{G}^{-H}\mathbf{D},
 \end{align}
 where $\mathbf{D}$ is a diagonal matrix to ensure power normalization at the base station, i.e., $\mathbf{D}^2=\diag \{\frac{1}{(\mathbf{G}^{-1}\mathbf{G}^{-H})_{1,1}}, \cdots, \frac{1}{(\mathbf{G}^{-1}\mathbf{G}^{-H})_{M,M}}\}$, where $(\mathbf{A})_{m,m}$ denotes the $m$-th element on the main diagonal of $\mathbf{A}$. As a result, the transmission power at the base station can be constrained as follows:
 \begin{align}
 & {\rm tr}\left\{\mathbf{P}\mathbf{P}^H\right\}\rho={\rm tr}\left\{ \mathbf{G}^{-H}\mathbf{D} \mathbf{D}^H\mathbf{G}^{-1}\right\}\rho ={\rm tr}\left\{ \mathbf{G}^{-1}\mathbf{G}^{-H}\mathbf{D}^2\right\}\rho=M\rho,
 \end{align}
 where $\rho$ denotes the transmit signal-to-noise ratio (SNR).

  With the design in \eqref{constraint 3} and \eqref{p design}, the signal model for  user  $m$ can now be written as follows:
   \begin{align}
 \mathbf{v}_m^H\mathbf{y}_m &=   \frac{\mathbf{g}_m^H}{\sqrt{L(d_m)}} \mathbf{p}_m (\alpha_{m}s_{m} + \alpha_{m'}s_{m'})  + \sum_{i\neq m}\frac{\mathbf{g}_m^H}{\sqrt{L(d_m)}} \mathbf{p}_i (\alpha_{i}s_{i} + \alpha_{i'}s_{i'}) +\mathbf{v}_m^H(\mathbf{w}_{I_m}+\mathbf{n}_m)\\ \nonumber &=   \frac{ (\alpha_{m}s_{m} + \alpha_{m'}s_{m'}) }{\sqrt{(L(d_m))(\mathbf{G}^{-1}\mathbf{G}^{-H})_{m,m}}} +\mathbf{v}_m^H(\mathbf{w}_{I_m}+\mathbf{n}_m).
 \end{align}

 For notational simplicity, we define $y_m=\mathbf{v}_m^H\mathbf{y}_m$, $h_m= \frac{1 }{\sqrt{L(d_m)(\mathbf{G}^{-1}\mathbf{G}^{-H})_{m,m}}} $, $ w_{I_m} = \mathbf{v}_m^H\mathbf{w}_{I_m}$, and $n_m=\mathbf{v}_m^H\mathbf{n}_m$. Therefore, the use of the signal alignment based precoding and detection matrices   decomposes the multi-user MIMO-NOMA channels into $M$ pairs of single-antenna NOMA channels. In particular, within each pair, the two users receive the following scalar observations
 \begin{align}
 y_m = h_m  (\alpha_{m}s_{m} + \alpha_{m'}s_{m'}) +w_{I_{m}}+n_m ,
 \end{align}
 and
 \begin{align}
 y_{m'} = h_{m'}  (\alpha_{m}s_{m} + \alpha_{m'}s_{m'})+w_{I_{m'}} +n_{m'},
 \end{align}
 where $y_{m'}$ and $n_{m'}$ are defined similar to $y_m$ and $n_m$, respectively. Note that $ h_{m'}= \frac{1 }{\sqrt{L(d_{m'})(\mathbf{G}^{-1}\mathbf{G}^{-H})_{m,m}}} $, and it is important to point out that $h_m$ and $h_{m'}$ share the same small scale fading gain with different distances.

 Recall that   two users belonging to the same pair are selected from $\mathcal{D}_1$ and $\mathcal{D}_2$, respectively,  which means that $d_m<d_{m'}$. Therefore, the two users from the same pair are ordered without any ambiguity, which simplifies the design of the power allocation coefficients, i.e., $\alpha_m\leq \alpha_{m'}$, following the NOMA principle. User $m'$   decodes its message with the following signal-to-interference-plus-noise ratio (SINR)
 \begin{align}
 SINR_{m'} =\frac{\rho|h_{m'}|^2\alpha_{m'}^2}{\rho|h_{m'}|^2\alpha_{m}^2+|\mathbf{v}_{m'}|^2+
 |\mathbf{v}_{m'}^H\mathbf{1}_N|^2
 I_{m'}},
 \end{align}
 where   the interference term  is given by
  \begin{align}
 I_{m'}= \underset{j\in\Psi_I}{\sum}\frac{\rho_I}{L\left(d_{I_j,m'}\right)},
  \end{align}
 User $m$   carries  out successive interference cancellation (SIC)  by first removing the message to user $m'$ with   SINR, $SINR_{m,m'} = \frac{\rho|h_{m}|^2\alpha_{m'}^2}{\rho|h_{m}|^2\alpha_{m}^2+|\mathbf{v}_m|^2+
 |\mathbf{v}_m^H\mathbf{1}_N|^2
 I_{m'}}$, and then decoding  its own message  with SINR
  \begin{align}
 SINR_{m} = \frac{\rho |h_{m}|^2\alpha_{m}^2}{|\mathbf{v}_m|^2+
 |\mathbf{v}_m^H\mathbf{1}_N|^2
 I_{m}} .
 \end{align}
 which becomes the SNR if $\rho_I=0$.

 \subsection{Uplink MIMO-NOMA transmission}
For the NOMA uplink case, user $m$ will send out an information bearing message $s_m$, and the signal transmitted by this user is denoted by $\alpha_m \mathbf{v}_ms_m$. Because of the reciprocity between uplink and downlink channels,   $\mathbf{v}_m$ which was used as a downlink detection vector can be used as a  precoding vector for the uplink scenario. Similarly $\mathbf{P}$ will be used as the detection matrix for the uplink case.   In this paper, we assume  that  the total transmission power from one user pair is normalized as follows:
 \begin{align}\alpha_m^2 |\mathbf{v}_m|^2+\alpha_{m'}^2 |\mathbf{v}_{m'}|^2\leq2\rho.
 \end{align}

The base station observes the following signal:
 \begin{align}
 \mathbf{y}_{BS}& = \sum^{M}_{m=1}\left(\frac{\mathbf{G}_m^H\alpha_m \mathbf{v}_ms_m}{\sqrt{L(d_m)}}   +\frac{\mathbf{G}_{m'}^H \alpha_{m'} \mathbf{v}_{m'}s_{m'}}{\sqrt{L(d_{m'})}}   \right) +\mathbf{w}_{I} +\mathbf{n}_{BS},
 \end{align}
 where $\mathbf{w}_{I} $ is the interference term defined as follows
  \begin{align}
 \mathbf{w}_{I}\triangleq \underset{j\in\Psi_I}{\sum}\frac{\sqrt{\rho_I}}{\sqrt{L\left(d_{I_j,BS}\right)}}\mathbf{1}_{M},
  \end{align}
$d_{I_j,BS}$ denotes the distance between the base station and the $j$-th interferer,   and the noise term is defined similarly as  in the previous section.
   The base station   applies a detection matrix $\mathbf{P}$ to its observations and the system model at the base station can be written as follows:
  \begin{align}\nonumber
 \mathbf{P}^H\mathbf{y}_{BS} &= \mathbf{P}^H\sum^{M}_{m=1}\left(\frac{\mathbf{G}_m^H\alpha_m \mathbf{v}_ms_m}{\sqrt{L(d_{m})}}   +\frac{ \mathbf{G}_{m'}^H \alpha_{m'} \mathbf{v}_{m'}s_{m'}}{\sqrt{L(d_{m'})}}   \right) +\mathbf{P}^H(\mathbf{w}_{I} +\mathbf{n}_{BS}).
 \end{align}

 As a result, the symbols from the $m$-th user pair can be detected based on
  {\small \begin{align}\nonumber
 \mathbf{p}_m^H\mathbf{y}_{BS} = \mathbf{p}_m^H \left(\frac{\mathbf{G}_m^H\alpha_m \mathbf{v}_ms_m}{\sqrt{L(d_{m})}}   +\frac{ \mathbf{G}_{m'}^H \alpha_{m'} \mathbf{v}_{m'}s_{m'}}{\sqrt{L(d_{m'})}}   \right) +\underset{{\rm interference(including~inter-pair~interference)~+~noise}}{\underbrace{\mathbf{p}_m^H\sum_{i\neq m}\left(\frac{\mathbf{G}_i^H\alpha_i \mathbf{v}_is_i}{\sqrt{L(d_{i})}}   +\frac{ \mathbf{G}_{i'}^H \alpha_{i'} \mathbf{v}_{i'}s_{i'}}{\sqrt{L(d_{i'})}}   \right) +\mathbf{p}_m^H(\mathbf{w}_{I}+\mathbf{n}_{BS})}}.
 \end{align}}
 In order to avoid   inter-pair interference,   the following constraint needs to be met
 \begin{align}
 \mathbf{p}_m^H\sum_{i\neq m}\left(\frac{\mathbf{G}_i^H\alpha_i \mathbf{v}_is_i}{\sqrt{L(d_{i})}}   +\frac{ \mathbf{G}_{i'}^H \alpha_{i'} \mathbf{v}_{i'}s_{i'}}{\sqrt{L(d_{i'})}}   \right)  =0,~\forall m\neq i.
 \end{align}

   Applying again the concept of signal alignment, the constraint that $ \mathbf{G}_{m}^H   \mathbf{v}_{m} =\mathbf{G}_{m'}^H  \mathbf{v}_{m'}$ is imposed  on the precoding vectors $\mathbf{v}_m$.
 Therefore, the same design of $\mathbf{v}_m$ as shown in \eqref{dtection 1} can be used. The total transmission power within one pair is given by
 \begin{align}
  &\rho\alpha_m^2 |\mathbf{v}_m|^2+\rho\alpha_{m'}^2 |\mathbf{v}_{m'}|^2  \leq\rho\max(\alpha_m^2,\alpha_{m'}^2)(|\mathbf{v}_m|^2+ |\mathbf{v}_{m'}|^2) \leq 2\rho.
\end{align}
 Therefore, the use of the precoding vector in \eqref{dtection 1}   ensures that  the total transmission power of one user pair is constrained.

  Applying the detection matrix defined in \eqref{p design}, the system model for the base station to decode the messages from the $m$-th pair can be written as follows:
   \begin{align}\label{uplink 1}
 y_{BS,m}=    h_m\alpha_m s_m   +h_{m'} \alpha_{m'} s_{m'}   + w_{BS,m}+n_{BS,m},
  \end{align}
  where $y_{BS,m}=\mathbf{p}_m^H\mathbf{y}_{BS}$, $w_{BS,m}=\mathbf{p}_m^H\mathbf{w}_{I}$, and $n_{BS,m}=\mathbf{p}_m^H\mathbf{n}_{BS}$.
  Therefore,  using the proposed precoding and detection matrices, we can decompose the multi-user MIMO-NOMA uplink channel into $M$ orthogonal  single-antenna NOMA channels.
  Note that the variance of the noise is normalized  as illustrated in the following:
  \begin{align}\label{vairance of noise}
  &\mathcal{E}\{\mathbf{p}_m^H\mathbf{n}_{BS}\mathbf{n}_{BS}^H\mathbf{p}_m\} = \mathbf{p}_m^H \mathbf{p}_m = (\mathbf{P}^H\mathbf{P})_{m,m}= (\mathbf{D}^H\mathbf{G}^{-1}\mathbf{G}^{-H}\mathbf{D})_{m,m}=  \frac{( \mathbf{G}^{-1}\mathbf{G}^{-H} )_{m,m}}{(\mathbf{G}^{-1}\mathbf{G}^{-H})_{m,m}}=1.
    \end{align}
    The SIC strategy can be applied to decode the users' messages, following   steps similar to those used in  the downlink scenario.

  \section{Performance Analysis for Downlink MIMO-NOMA Transmission}\label{section downlink}
Two types of power allocation policies are considered in this section. One is fixed power allocation and the other is inspired by the cognitive ratio concept, as  illustrated in the following two subsections, respectively. Recall that the precoding vectors $\mathbf{v}_m$ and $\mathbf{v}_{m'}$ are determined by $\mathbf{x}_m$ as shown in \eqref{dtection 1}. In this section,  a random choice of $\mathbf{x}_m$ is considered first. How to find a more sophisticated choice for $\mathbf{x}_m$ is investigated  in Section \ref{section extension}.

 \subsection{Fixed Power Allocation}
 In this case, the power allocation coefficients $\alpha_m$ and $\alpha_{m'}$ are constant and not related to the instantaneous realizations of the fading channels.  We will first focus on the outage performance of user $m'$. The outage probability of  user $m'$ to decode its   information is given by
 \begin{align}
&\mathrm{P}^o_{m'}= \mathrm{P}\left(\log\left(1+\frac{\rho|h_{m'}|^2\alpha_{m'}^2}{\rho|h_{m'}|^2
 \alpha_{m}^2+|\mathbf{v}_{m'}|^2+|\mathbf{v}_{m'}^H\mathbf{1}_N|^2I_{m'}}\right)< R_{m'}\right),
 \end{align}
 where $\mathrm{P}(x<a)$ denotes the probability for the event $x<a$. 
 The correlation between $\mathbf{v}_{m'}$ and $h_{m'}$ makes the evaluation of the above outage probability very challenging. Hence,  we  focus on the following modified expression for the outage probability
\begin{align}\nonumber
&\tilde{\mathrm{P}}_{m'}=  \mathrm{P}\left(\log\left(1+\frac{\rho|h_{m'}|^2\alpha_{m'}^2}{\rho|h_{m'}|^2
 \alpha_{m}^2+2+2\delta I_{m'}}\right)< R_{m'}\right).
\end{align}
Since $|\mathbf{v}_{m'}|^2+|\mathbf{v}_{m}|^2=2$, we have $|\mathbf{v}_{m'}|^2\leq2$ and $|\mathbf{v}_{m}|^2\leq2$. In addition, because $(\frac{1}{N}\sum_{n=1}^Nx_n)^2\leq \frac{1}{N}\sum^{N}_{n=1}x_n^2$, $|\mathbf{v}_{m'}^H\mathbf{1}_N|^2\leq N|\mathbf{v}_{m'}|^2$. Therefore, we have
\begin{align}
 {\mathrm{P}}_{m'}^o\leq \tilde{\mathrm{P}}_{m'},
\end{align}
for $\delta\geq N$, which means that $\tilde{\mathrm{P}}_{m'}$ provides an upper bound on $ {\mathrm{P}}_{m'}$ if $\delta \geq N$. Note that when $\delta=1$, the difference between $\tilde{\mathrm{P}}_{m'}$ and ${\mathrm{P}}_{m'}$ is very small as can be observed  from Fig. \ref{bound}, i.e., a choice of $\delta=1$ is sufficient to ensure that  $\tilde{\mathrm{P}}_{m'}$  provides a very tight approximation to  ${\mathrm{P}}_{m'}$. In addition, the use of $\tilde{\mathrm{P}}_{m'}$ will be sufficient to identify the achievable diversity order of the proposed MIMO-NOMA scheme.

\begin{figure}[!htbp]\centering\vspace{-1em}
    \epsfig{file=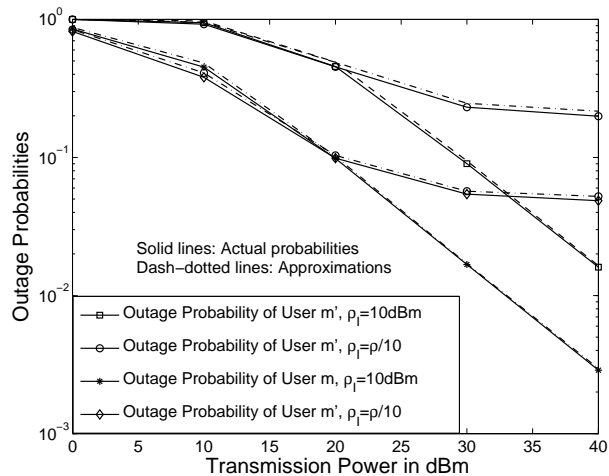, width=0.46\textwidth, clip=}
\caption{ Comparison between $\tilde{\mathrm{P}}_{i'}$ and $ {\mathrm{P}}_{i'}^o$, $i\in \{m,m'\}$. $R_m=R_{m'}=1.5$ bit per channel use (BPCU). $\lambda_I=10^{-4}$. $r=20$m and $r_1=10$m. $r_0=1$m and $\alpha_{m'}=\frac{3}{4}$. $M=N=2$.  The path loss exponent is $\alpha=3$, and the noise power is $-30$dBm.    }\label{bound}\vspace{-1em}
\end{figure}

 Given a random choice of $\mathbf{x}_m$, the following lemma provides an exact expression for $\tilde{\mathrm{P}}_{m'}$ as well as its high SNR approximation.
\begin{lemma}\label{lemma1}
If  $ \alpha_{m'}^2 \leq \alpha_{m}^2 \epsilon_{m'}$,   the probability $\tilde{\mathrm{P}}_{m'}=1$, where $\epsilon_{m'}=2^{R_{m'}}-1$. Otherwise the probability $\tilde{\mathrm{P}}_{m'}$ can be expressed as follows:
\begin{align}
\tilde{\mathrm{P}}_{m'}&=1-  \frac{2}{  r^2-  r_1^2} \int_{r_1}^r e^{- 2\phi_{m'}x^\alpha}\varphi_I(x)xdx,
\end{align}
where $\phi_{m'}= \frac{\epsilon_{m'}}{\rho\alpha_{m'}^2 - \rho\alpha_{m}^2 \epsilon_{m'}}$,   $
\varphi_I(x) = e^{-\pi \lambda_I (\beta_{m'}(x))^{\frac{2}{\alpha}}
\gamma\left(\frac{1}{\alpha},\frac{\beta_{m'}(x)}{r_0^\alpha}\right)}$,  $\beta_{m'}(x)=2\phi_{m'}\delta \rho_I L\left(x^\alpha\right)$, and $\gamma(\cdot)$ denotes the incomplete Gamma function.

If $\rho_I$ is fixed and transmit SNR $\rho$ approaches infinity, the outage probability can be approximated as follows:
\begin{align}
\tilde{\mathrm{P}}_{m'}&\approx\frac{2\phi_{m'}(2+\tilde{\theta}_{m'})}{  r^2-  r_1^2}\frac{\left(r^{\alpha+2}-r_1^{\alpha+2}\right)}{\alpha+2},
\end{align}
where $\tilde{\theta}_{m'}=2\pi \lambda_I\delta \rho_I\frac{\alpha}{r_0}$.
For the special case of $\rho_I=0$,  $\tilde{\mathrm{P}}_{m'}$ simplifies to
{\small \begin{align}
&\tilde{\mathrm{P}}_{m'}   = 1 - \frac{1}{  r^2-  r_1^2} \left(  e^{-2\phi_{m'}r^{\alpha}} r^2 - e^{-2\phi_{m'}r_1^{\alpha}} r_1^2\right)-\frac{(2\phi_{m'})^{-\frac{2}{\alpha}}}{  r^2-  r_1^2}\left(\gamma\left(\frac{2}{\alpha}+1, 2\phi_{m'}r^{\alpha}\right)-\gamma\left(\frac{2}{\alpha}+1, 2\phi_{m'}r_1^{\alpha}\right)\right) .
\end{align}}
\end{lemma}
\begin{proof}
Please refer to  Appendix A.
\end{proof}
By using the high SNR approximation obtained in Lemma \ref{lemma1} and also  the fact that  both $\phi_{m'}$ and $\theta_{m'}$ are at the order of $\frac{1}{\rho}$, the achievable diversity gain is obtained   in the following corollary.
\begin{corollary}
If  $ \alpha_{m'}^2 > \alpha_{m}^2 \epsilon_{m'}$, the diversity order achieved by the proposed MIMO-NOMA framework for  user $m'$ is one.
\end{corollary}

On the other hand,  user $m$   first decodes the message for user $m'$ before decoding its own message via SIC. Therefore, the outage probability at  user $m$  is given by
 \begin{align}
&\mathrm{P}^o_{m}= \mathrm{P}\left(\log\left(1+\frac{\rho|h_{m}|^2\alpha_{m'}^2}{\rho|h_{m}|^2
 \alpha_{m}^2+|\mathbf{v}_m|^2+|\mathbf{v}_m^H\mathrm{1}_N|^2I_{m}}\right)< R_{m'}\right)\\ \nonumber &+\mathrm{P}\left(\log\left(1+\frac{\rho|h_{m}|^2\alpha_{m}^2}{|\mathbf{v}_m|^2+|\mathbf{v}_m^H\mathrm{1}_N|^2I_{m}}\right)< R_{m},\log\left(1+\frac{\rho|h_{m}|^2\alpha_{m'}^2}{\rho|h_{m}|^2
 \alpha_{m}^2+|\mathbf{v}_m|^2+|\mathbf{v}_m^H\mathrm{1}_N|^2I_{m}}\right)> R_{m'}\right).
 \end{align}
Again,  we   focus on a modified  expression for  the outage probability as follows:
\begin{align}\label{outage mm}
\tilde{\mathrm{P}}_{m}&= \mathrm{P}\left(\log\left(1+\frac{\rho|h_{m}|^2\alpha_{m'}^2}{\rho|h_{m}|^2
 \alpha_{m}^2+2+2\delta  I_{m}}\right)< R_{m'}\right)\\ \nonumber &+\mathrm{P}\left(\log\left(1+\frac{\rho|h_{m}|^2\alpha_{m}^2}{2+2\delta I_{m}}\right)< R_{m},\log\left(1+\frac{\rho|h_{m}|^2\alpha_{m'}^2}{\rho|h_{m}|^2
 \alpha_{m}^2+2+2\delta I_{m}}\right)> R_{m'}\right),
 \end{align}
 which is   an upper bound for $\delta \geq N$ as explained in the proof for Lemma~\ref{lemma2}.   Fig. \ref{bound} demonstrates that $\tilde{\mathrm{P}}_{m}$ with a choice of $\delta=1$ yields a tight upper approximation  on $ {\mathrm{P}}_{m}$.
The following lemma provides an exact expression for this  probability as well as its high SNR approximation.
\begin{lemma}\label{lemma2}
If  $ \alpha_{m'}^2 \leq \alpha_{m}^2 \epsilon_{m'}$,  the probability $\tilde{\mathrm{P}}_{m}=1$, otherwise the probability $\tilde{\mathrm{P}}_{m'}$ can be expressed as follows:
\begin{align}
\tilde{\mathrm{P}}_{m}&=1-  \frac{2}{   r_1^2} \int_{0}^{r_0} e^{- 2\tilde{\phi}_{m}r_0^\alpha}\varphi_I(r_0)xdx- \frac{2}{   r_1^2} \int_{r_0}^{r_1} e^{- 2\tilde{\phi}_{m}x^\alpha}\varphi_I(x)xdx,
\end{align}
where $\tilde{\phi}_m=\max\{\phi_m,\phi_{m'}\}$ and $\phi_m=\frac{\epsilon_m}{\rho \alpha^2_{m}}$.
If  $\rho_I$ is fixed and the transmit SNR $\rho$ approaches infinity, the outage probability can be approximated as follows:
\begin{align}
\tilde{\mathrm{P}}_{m} &\approx \frac{\tilde{\phi}_m(2+\tilde{\theta}_{m'})}{r_1^2(\alpha+2) }\left(\alpha r_0^{\alpha+2}+2r_1^{\alpha+2}\right),
\end{align}
where $\tilde{\theta}_{m'}$ was defined in Lemma \ref{lemma1}.
\end{lemma}
\begin{proof}
Please refer to Appendix B.
\end{proof}

%

\subsection{Cognitive Radio Power Allocation }\label{cognitive section}
In this section, a cognitive radio inspired power allocation strategy  is studied. In particular, assume  that user $m'$ is viewed as a primary user in a cognitive ratio network. With orthogonal multiple access, the bandwidth resource occupied by  user $m'$ cannot be reused by   other users, despite its poor channel conditions. In contrast, with NOMA, one additional user, i.e.,  user $m$, can be served simultaneously, under the condition that the QoS requirements of  user $m'$ can still be met.

 In particular, assume  that user $m'$ needs to achieve a target data rate of $R_{m'}$, which means that the power allocation coefficients of NOMA need to satisfy the following constraint
\begin{align}\label{constraint cg}
\frac{\rho|h_{m'}|^2\alpha_{m'}^2}{\rho|h_{m'}|^2
 \alpha_{m}^2+|\mathbf{v}_{m'}|^2+|\mathbf{v}_{m'}^H\mathbf{1}_N|^2I_{m'}}>\epsilon_{m'},
 \end{align}
 which leads to the following choice for $\alpha_m$
 \begin{align}\label{choice 1 cg}
 \alpha_m^2 = \max\left(0, \frac{\rho |h_{m'}|^2-\epsilon_{m'}(|\mathbf{v}_{m'}|^2+|\mathbf{v}_{m'}^H\mathbf{1}_N|^2I_{m'})}{(1+\epsilon_{m'})\rho |h_{m'}|^2}\right).
 \end{align}
 It is straightforward  to show that  $\frac{\rho |h_{m'}|^2-\epsilon_{m'}(|\mathbf{v}_{m'}|^2+|\mathbf{v}_{m'}^H\mathbf{1}_N|^2I_{m'})}{(1+\epsilon_{m'})\rho |h_{m'}|^2}$ is always less than one.

An outage   at user $m'$ means here that all   power is allocated to  user $m'$, but outage still occurs. As a result, the outage probability of  user $m'$ is exactly the same as that in conventional orthogonal MA systems. Therefore, in this section, we only focus on the outage  probability of  user $m$  which can be expressed as follows:
 \begin{align}
\mathrm{P}^o_{m}=&\mathrm{P}\left(|h_m|^2<\max\left\{ \phi_{m'} (|\mathbf{v}_m|^2+|\mathbf{v}_m^H\mathbf{1}_N|^2I_{m}), \phi_m ( |\mathbf{v}_m|^2+|\mathbf{v}_m^H\mathbf{1}_N|^2I_{m})\right\}\right),
 \end{align}
 if $ \alpha_{m'}^2 > \alpha_{m}^2 \epsilon_{m'}$, otherwise outage always occurs. It can be verified that $\alpha_{m'}^2 \leq \alpha_{m}^2 \epsilon_{m'}$ is equivalent to $\alpha_m=0$, in the context of  cognitive radio power allocation.

 Analyzing this outage probability is very difficult due to the following two reasons. First, $h_m$ and $\mathbf{v}_m$ are correlated, and second,  the users experience different but correlated co-channel interference, i.e., $I_m\neq I_{m'}$. Therefore, in this subsection, we only focus on the case without co-channel interference, i.e., $\rho_I=0$. In particular,  we focus on the following outage probability
  \begin{align}
\tilde{\mathrm{P}}_{m}=&\mathrm{P}\left(|h_m|^2<2\max\left\{ \bar{\phi}_{m'}, \bar{\phi}_m \right\}\right),
 \end{align}
 where $\bar{\phi}_m=\frac{\epsilon_m}{\rho \bar{\alpha}^2_{m}}$,   $\bar{\phi}_{m'}= \frac{\epsilon_{m'}}{\rho\bar{\alpha}_{m'}^2 - \rho\bar{\alpha}_{m}^2 \epsilon_{m'}}$,  and \begin{align}
  \bar{\alpha}_m^2 = \max\left(0, \frac{\rho |h_{m'}|^2-2\epsilon_{m'}}{(1+\epsilon_{m'})\rho |h_{m'}|^2}\right).
 \end{align}
Similarly to the case with fixed power allocation,  the outage probability $\tilde{\mathrm{P}}_{m}$ tightly bounds  $ {\mathrm{P}}^o_{m}$. The following lemma provides the expression for the outage probability $\tilde{\mathrm{P}}_{m}$.
\begin{lemma}\label{lemma 4}
When $\rho_I=0$, the outage probability can be expressed as follows:
\begin{align}\label{eq lemma 3}
\tilde{\mathrm{P}}_{m} &= 1 - \Upsilon_1\left(\frac{2\epsilon_{m'}}{\rho}\right)\Upsilon_2
\left(\frac{2\epsilon_m(1+\epsilon_{m'}) }{\rho }\right),
\end{align}
where
\begin{align}
\Upsilon_1(y) & =  \frac{1}{  r^2-  r_1^2} \left(  e^{-yr^{\alpha}} r^2 - e^{-yr_1^{\alpha}} r_1^2\right)+\frac{y^{-\frac{2}{\alpha}}}{  r^2-  r_1^2}\left(\gamma\left(\frac{2}{\alpha}+1, yr^{\alpha}\right)-\gamma\left(\frac{2}{\alpha}+1,yr_1^{\alpha}\right)\right) .
\end{align}
and
\begin{align}\nonumber
\Upsilon_2(z) & = \frac{r_0^2e^{- zr_0^\alpha}}{r_1^2}+ \frac{1}{    r_1^2}  \left(  e^{-zr_1^{\alpha}} r_1^2 -  e^{-zr_0^{\alpha}} r_0^2\right) +\frac{z^{-\frac{2}{\alpha}}}{    r_1^2} \left(\gamma\left(\frac{2}{\alpha}+1, zr_1^{\alpha}\right)-\gamma\left(\frac{2}{\alpha}+1, zr_0^{\alpha}\right)\right) .
\end{align}
At high SNR, the outage probability can be approximated as follows:
\begin{align}
\tilde{\mathrm{P}}_{m} &\approx  \frac{4\epsilon_{m'}}{ \rho(2+\alpha)( r^2-  r_1^2)} \left(  r^{\alpha+2} -r_1^{\alpha+2} \right)  +  \frac{2r_0^{2+\alpha} \epsilon_m(1+\epsilon_{m'})}{\rho r_1^2}  +\frac{4\epsilon_m(1+\epsilon_{m'}) }{\rho  (2+\alpha)  r_1^2} \left(  r_1^{\alpha+2} -  r_0^{\alpha+2}\right).
\end{align}
\end{lemma}
\begin{proof}
Please refer to Appendix C.
\end{proof}
{\it Remark 1}: By using the above lemma, it is straightforward to show that a diversity gain of one is still achievable at  user $m$ (i.e., there is no error floor), and it is important to point out that this is achieved when  user $m'$ experiences the same outage performance as if it solely uses the channel. Therefore, by using the proposed cognitive radio NOMA, one additional user, user $m$, is introduced into the system to share the spectrum with the primary user, user $m'$,  without causing any performance degradation at user $m'$.

{\it Remark 2}: For the above cognitive radio NOMA scheme, it was assumed that the message for  user $m'$ is decoded first at both receivers. Nevertheless, different SIC decoding strategies can be used, and their impact can be obtained in a straightforward manner  from the analysis in the next section, where  more complicated uplink transmission schemes are studied. It is worth pointing out that  $\alpha_m^2$ in \eqref{choice 1 cg} is always  smaller than $\frac{1}{2}$, for $ R_{m'}\geq 1$. For example, when $\alpha^2_m=0$, the inequality $ \alpha_m^2 -\frac{1}{2}<0$ holds obviously. When $\alpha_m^2>0$,
\begin{align}
 \alpha_m^2 -\frac{1}{2}&=  \frac{\rho |h_{m'}|^2-\epsilon_{m'}(|\mathbf{v}_m|^2+|\mathbf{v}_m^H\mathbf{1}_N|^2I_{m'})}{(1+\epsilon_{m'})\rho |h_{m'}|^2} -\frac{1}{2}=  \frac{\rho |h_{m'}|^2(1-\epsilon_{m'})-2\epsilon_{m'}(|\mathbf{v}_m|^2+
 |\mathbf{v}_m^H\mathbf{1}_N|^2I_{m'})}
 {2(1+\epsilon_{m'})\rho |h_{m'}|^2}  \leq 0,
\end{align}
if $R_{m'}\geq 1$.

\subsection{Selection of the User Detection Vectors} \label{section extension}
Previously, a random choice of $\mathbf{v}_m$ and $\mathbf{v}_{m'}$ has been used and analyzed. In the case of $2N-M>1$, there is more than one possible choice based on the null space, $\mathbf{U}_m$, defined in \eqref{dtection 1}. In this section, we study how to utilize these additional degrees of freedom and  analyze their  impact on the outage probability.

Finding the optimal choice for $\mathbf{v}_m$ and $\mathbf{v}_{m'}$ is challenging, since the choice of the detection vectors for one user pair has an impact on those of the other user pairs. For example, the choice of  $\mathbf{v}_m$ and $\mathbf{v}_{m'}$ will affect   the $m$-th column of the effective fading matrix $\mathbf{G}$. Recall that the data rates of the users from the $i$-th pair is a function of $\frac{1}{( {\mathbf{G}}^{-1} {\mathbf{G}}^{-H})_{i,i}}$. Therefore, the detection vector chosen by the $m$-th user pair will also affect the  data rates of the users in the $i$-th pair, $m\neq i$.

In order to avoid this tangled effect, a simple algorithm for     detection vector selection is proposed in Table \ref{alg:stuff}. The following lemma shows the diversity gain achieved by the proposed selection algorithm.

\begin{algorithm}[t]
\caption{The selection of the detection vectors $\mathbf{v}_m$ and $\mathbf{v}_{m'}$}
\label{alg:stuff}
\begin{algorithmic}[1]
\FOR{$i=1$ to $(2N-M)$}
\STATE Set $\mathbf{x}_{m,i}=\begin{bmatrix} \mathbf{0}_{1\times (i-1)} &1&\mathbf{0}_{1\times (M-i)} \end{bmatrix}^H$,   $\forall m\in\{1,\cdots,M\}$.
\STATE Choose the detection vector as $\begin{bmatrix}  \mathbf{v}_{m,i}^H &\mathbf{v}_{m',i}^H\end{bmatrix}^H=\mathbf{U}_m\mathbf{x}_{m,i}$ and determine   vector $ {\mathbf{g}}_{m,i}=\mathbf{G}_m^H\mathbf{v}_{m,i}$.
\STATE Construct the effective small scale fading matrix, denoted by $\bar{\mathbf{G}}_i$, by using $\mathbf{g}_{m,i}$, i.e., $\bar{\mathbf{G}}_i=\begin{bmatrix}\mathbf{g}_{1,i}& \cdots&\mathbf{g}_{M,i}  \end{bmatrix}^H$
\STATE Find the effective small scale fading gain for each user pair, $\gamma_{m,i}=\frac{1}{(\bar{\mathbf{G}}_i^{-1}\bar{\mathbf{G}}_i^{-H})_{m,m}}$.
\STATE Find the smallest fading gain, $\gamma_{\min,i}=\min\{\gamma_{1,i},\cdots,\gamma_{M,i}\}$.
\ENDFOR
\STATE Find the index $i$ which maximizes the smallest fading gain, $i^*=\underset{i\in{\{1, \cdots, 2N-M\}}}{\arg} \max ~ \gamma_{\min,i}$ .
\end{algorithmic}
\end{algorithm}\vspace{-1em}

\begin{lemma}\label{lemma 5}
Consider the use of a fixed set of power allocation coefficients. If  $ \alpha_{m'}^2 \leq \alpha_{m}^2 \epsilon_{m'}$,  the probability $\tilde{\mathrm{P}}_{m'}=1$, otherwise the use of the algorithm proposed in Table \ref{alg:stuff}   ensures that a diversity gain of $(2N-M)$ is achieved.
\end{lemma}
\begin{proof}
Please refer to Appendix D.
\end{proof}
As can be seen from Lemma \ref{lemma 5}, the use of the proposed selection algorithm can increase the diversity gain from $1$ to $(2N-M)$, which is a significant improvment compared to the scheme in \cite{Zhiguo_mimoconoma}. Consider a scenario with $N=M$ as an example. The proposed scheme can achieve a diversity gain of $M$, whereas the one in \cite{Zhiguo_mimoconoma} can only achieve a diversity gain of $1$, for an unordered user. Note, however,  that the scheme in \cite{Zhiguo_mimoconoma} does not require CSI at the transmitter.

\section{Performance Analysis of MIMO-NOMA Uplink Transmission}\label{section uplink}
Because of the symmetry between the uplink and downlink system models of Section \ref{section system model},  in this section,  we only focus on the difference between two scenarios.
 One important  observation for uplink NOMA is that the sum rate is always the same, no matter which decoding order is used.  Therefore, in this section, we first analyze the outage probability with respect to the sum rate for a fixed power allocation. The use of a randomly selected $\mathbf{x}_m$ is considered  in order to obtain  tractable analytical results.

\subsection{Fixed Power Allocation}
Recall that, if the message from  user $m$ is decoded first, the base station can correctly decode the message with  rate
\begin{align}
  R_{m,BS,I} = \log\left(1+\frac{\rho|h_m|^2\alpha^2_m}{\rho|h_{m'}|^2\alpha^2_{m'}+I_{BS,m}+1}\right),
  \end{align}
where    the interference power is given by
  \begin{align}
 I_{BS,m}= \underset{j\in\Psi_I}{\sum}\frac{\rho_I|\mathbf{p}_m^H
 \mathbf{1}_M|^2}{L\left(d_{I_j,BS}\right)}.
  \end{align}

After subtracting the message from  user $m$, the base station can decode the message from user $m'$ correctly with the following rate
  \begin{align}
  R_{m',BS,I} = \log \left(1+\frac{\rho|h_{m'}|^2\alpha^2_{m'}}{I_{BS,m}+1}\right).
  \end{align}
Therefore, the sum rate achieved by NOMA in the $m$-th sub-channel is given by
  \begin{align}
  R_s &=  R_{m,BS,I}+ R_{m',BS,I} =\log\left(1+\frac{\rho |h|_m^2\alpha_m^2+\rho |h|_{m'}^2\alpha_{m'}^2}{I_{BS,m}+1}\right).
  \end{align}
It is straightforward to verify that the exactly same sum rate is achieved if the message from  user $m'$ is decoded first. Therefore, the outage probability for the sum rate can be expressed as follows:
  \begin{align}
  \mathrm{P}_s = \mathrm{P}\left(R_s<R_m+R_{m'}\right).
  \end{align}
Note that the term for the interference power contains $|\mathbf{p}_m^H\mathbf{1}_M|^2$ which makes the calculation very difficult. Since $|\mathbf{p}_m^H\mathbf{1}_M|^2\leq M |\mathbf{p}_m^H|^2=M$,   we focus on the following modified expression of  the outage probability
  \begin{align}
  \tilde{\mathrm{P}}_s = \mathrm{P}\left(\log\left(1+\frac{\rho |h|_m^2\alpha_m^2+\rho |h|_{m'}^2\alpha_{m'}^2}{\delta I_m+1}\right)<R_m+R_{m'}\right),
  \end{align}
  where $I_{m}= \underset{j\in\Psi_I}{\sum}\frac{\rho_I}{L\left(d_{I_j,BS}^\alpha\right)}$. Similarly to the downlink case, $\tilde{\mathrm{P}}_s$ provides an upper bound on $ {\mathrm{P}}_s$ for $\delta  \geq M$. In the simulation section, we will demonstrate that $\tilde{\mathrm{P}}_s $ with a choice of $\delta=1$ provides a tight approximation to ${\mathrm{P}}_s $.

Define the small scale fading gain as $x\triangleq \frac{1}{ (\mathbf{G}^{-1}\mathbf{G}^{-H})_{m,m}} $. The  sum rate outage probability can be expressed as follows
  \begin{align}
 \tilde{\mathrm{P}}_s &= \mathrm{P}\left( \frac{\rho \frac{x}{L(d_m)}\alpha_m^2+\rho \frac{x}{L(d_{m'})}\alpha_{m'}^2}{\delta I_m+1} <\epsilon\right) = \mathrm{P}\left( x<\frac{\epsilon(\delta I_m+1)} { \frac{\rho \alpha_m^2}{L(d_m)}+ \frac{\rho \alpha_{m'}^2}{L(d_{m'})}}\right),
  \end{align}
  where $\epsilon=2^{R_m+R_{m'}}-1$.
Following the same steps  as in the proof of Lemma \ref{lemma1}, the above probability can be expressed as follows:
\begin{align}\label{uplink1}
\tilde{\mathrm{P}}_s&=1-  \frac{4}{r_1^2(  r^2-  r_1^2)} \int_{r_1}^r\int_{0}^{r_1} e^{- \zeta(x,y)} e^{-\pi \lambda_I (\rho_I\delta \zeta(x,y))^{\frac{2}{\alpha}}
\gamma\left(\frac{1}{\alpha},\frac{\rho_I\delta \zeta(x,y)}{r_0^\alpha}\right)}xdxydy,
\end{align}
where  $\zeta(d_m,d_{m'})=\frac{\epsilon}{ \frac{\rho \alpha_m^2}{L(d_m)}+ \frac{\rho \alpha_{m'}^2}{L(d_{m'})}}$.

In order to obtain some insights regarding  the above probability, we again consider the case that $\rho$ tends to infinity and $\rho_I$ is fixed.
Since both $d_m$ and $d_{m'}$ are bounded, $\zeta(d_m,d_{m'})$ approaches zero at high SNR. Therefore the above probability can be approximated as follows:
\begin{align}
\tilde{\mathrm{P}}_s&\approx 1-  \frac{4}{r_1^2(  r^2-  r_1^2)} \int_{r_1}^r\int_{0}^{r_1} e^{-  \zeta(x,y)} e^{-\pi \lambda_I (\rho_I\delta \zeta(x,y))^{\frac{2}{\alpha}}
\alpha\left(\frac{\rho_I\delta \zeta(x,y)}{r_0^\alpha}\right)^{\frac{1}{\alpha}}}xdxydy
\\ \nonumber &\approx 1-  \frac{4}{r_1^2(  r^2-  r_1^2)} \int_{r_1}^r\int_{0}^{r_1}  e^{- \zeta(x,y)
 \left(\frac{\pi \lambda_I\alpha\rho_I}{r_0}+1 \right)}xdxydy.
\end{align}
With some algebraic manipulations, the above probability can be simplified as follows:
\begin{align}
\tilde{\mathrm{P}}_s &\approx 1-  \frac{4}{r_1^2(  r^2-  r_1^2)} \int_{r_1}^r\int_{0}^{r_1} \left(1- \zeta(x,y)
 \left(\frac{\pi \lambda_I\delta \alpha\rho_I}{r_0}+1 \right)\right)xdxydy\\ \nonumber
 &\approx  \frac{4\left(\frac{\pi \lambda_I\delta \alpha\rho_I}{r_0}+1 \right)}{r_1^2(  r^2-  r_1^2)} \int_{r_1}^r\int_{0}^{r_1}  \zeta(x,y)
 xdxydy.
\end{align}
Therefore, the outage probability can be approximated as follows:
\begin{align}\label{uplink 2}
\tilde{\mathrm{P}}_{s} &\approx   \frac{4\xi\epsilon\left(\frac{\pi \lambda_I\delta \alpha\rho_I}{r_0}+1 \right)}{\rho r_1^2(  r^2-  r_1^2)} \sim \frac{1}{\rho},
\end{align}
where $\xi=\int_{r_1}^r\int_{0}^{r_1}  \frac{xy}{ \frac{  \alpha_m^2}{L(x)}+ \frac{  \alpha_{m'}^2}{L(y)}}
 dxdy$ is a constant and not related to the SNR. Hence, a diversity gain of $1$ is   achievable for the sum rate.

\subsection{Cognitive Radio Power Allocation}The design of  cognitive radio NOMA for uplink transmission is more complicated, as explained  in the following. To simplify the illustration,  we omit the interference term in this section, i.e., $\rho_I=0$. For downlink transmission, $\alpha_m^2<\frac{1}{2}$ was sufficient to decide the SIC decoding order. However, there are more uncertainties in the uplink case, since   $\alpha_{m'}^2|h_{m'}|^2$ is not necessarily larger than  $\alpha_m^2|h_m|^2$ even if $\alpha_{m'}^2>\frac{1}{2}$. Therefore, the base station can apply two types of decoding strategies, i.e., it may  decode the message from  user $m'$ first, or that of  user $m$ first. These strategies will yield different tradeoffs between the outage performance of the two users, as explained in the following subsections, respectively.

\subsubsection{Case I} When the message from  user $m'$ is decoded first, in order to guarantee the QoS at  user $m'$, we impose the following power constraint for the power allocation coefficients
 \begin{align}
 \log \left(1+ \frac{\rho|h_{m'}|^2\alpha^2_{m'}}{\rho|h_{m}|^2\alpha^2_{m}+1} \right)>R_{m'},
 \end{align}
 which leads to the following choice for $\alpha_{m'}$
 \begin{align}\label{cog 3}
\alpha^2_{m'} =\min\left\{1,\frac{\epsilon_{m'}+\rho\epsilon_{m'}|h_m|^2}{ \rho|h_{m'}|^2+\epsilon_{m'}\rho|h_{m}|^2}\right\}.
 \end{align}

Following the same steps as in the proof of Lemma \ref{lemma2}, the outage probability $\mathrm{P}^I_{m',BS}$ can be evaluated as follows:
\begin{align}\label{case i1}
\mathrm{P}^I_{m',BS} &= \mathrm{P}\left(\frac{\epsilon_{m'}+\rho\epsilon_{m'}|h_m|^2}{ \rho|h_{m'}|^2+\epsilon_{m'}\rho|h_{m}|^2}>1\right)\\ \nonumber &=\mathrm{P}\left(|h_{m'}|^2<\frac{\epsilon_{m'}}{\rho}\right) = 1 -\Upsilon_1\left(\frac{\epsilon_{m'}}{\rho}\right),
\end{align}
and following the same steps as in the proof of Lemma \ref{lemma 4}, the outage probability $\mathrm{P}^I_{m,BS}$ can be evaluated as follows:
 \begin{align}\label{case 12}
\mathrm{P}^I_{m,BS} =& \mathrm{P}\left(\frac{\epsilon_{m'}+\rho\epsilon_{m'}|h_m|^2}{ \rho|h_{m'}|^2+\epsilon_{m'}\rho|h_{m}|^2}>1\right)\\ \nonumber& +\mathrm{P}\left(\frac{\epsilon_{m'}+\rho\epsilon_{m'}|h_m|^2}{ \rho|h_{m'}|^2+\epsilon_{m'}\rho|h_{m}|^2}<1, \log\left(1+\rho |h_m|^2\frac{\rho|h_{m'}|^2-\epsilon_{m'}}{ \rho|h_{m'}|^2+\epsilon_{m'}\rho|h_{m}|^2}\right)<R_m \right)\\  \label{for compare} =&\mathrm{P}\left(|h_{m'}|^2<\frac{\epsilon_{m'}}{\rho}\right)+\mathrm{P}
\left(x>\frac{\epsilon_{m'}L(d_{m'})}{\rho},  x<L(d_m)\frac{\epsilon_m}{\rho} +L(d_{m'})\frac{\epsilon_{m'}}{\rho}\left(1+\epsilon_m\right)\right) \\ \nonumber=&  1 -\Upsilon_1\left(\frac{\epsilon_{m'}\left(1+\epsilon_m\right)}{\rho}\right)
\Upsilon_2\left(\frac{\epsilon_m}{\rho} \right).
\end{align}

\subsubsection{Case II} When the message from  user $m$ is decoded first, in order to guarantee the QoS at user $m'$, we impose the following power constraint for the power allocation coefficients
 \begin{align}
 \log \left(1+ \rho|h_{m'}|^2\alpha^2_{m'}\right)>R_{m'},
 \end{align}
 which leads to the following choice for $\alpha_{m'}$
 \begin{align}\label{cog 2}
\alpha^2_{m'} =\min\left\{1,\frac{\epsilon_{m'}}{ \rho|h_{m'}|^2}\right\}.
 \end{align}

 With this choice,   we can ensure that  the outage probabilities of both users are identical,  i.e., $ {\mathrm{P}}^{II}_{m,BS}= {\mathrm{P}}^{II}_{m',BS}$, as explained in the following. The outage events that occur at user $m'$ can be divided into the following three events
 \begin{itemize}
 \item $\tilde{E}_1$: All the power is allocated to  user $m'$, i.e., $\alpha_{m'}=1$, but the user is still in outage. The NOMA system is degraded to a scenario in which  only  user $m'$ is served.

 \item $\tilde{E}_2$: When $\alpha^2_{m'}<1$,   outage occurs at  user $m$, and   SIC is stopped.

\item $\tilde{E}_3$: When $\alpha^2_{m'}<1$, no outage occurs at  user $m$, but outage occurs at  user $m'$.

 \end{itemize}
 It is straightforward to show that $\tilde{E}_3$ will not happen, i.e., $\mathrm{P}(\tilde{E}_3)=0$. Therefore $ {\mathrm{P}}_{m',BS}=\mathrm{P}(\tilde{E}_1)+\mathrm{P}(\tilde{E}_2)$. On the other hand, there are only two outage events  for decoding  the message from user $m$, which are $\tilde{E}_1$ and $\tilde{E}_2$, respectively. Therefore, the outage probabilities of the two users are the same, $ {\mathrm{P}}_{m,BS}^{II}= {\mathrm{P}}^{II}_{m',BS}$.

Therefore, we only need to study  the outage probability for the message from  user $m$.  With the choice shown in \eqref{cog 2}, the outage probability can be rewritten as follows:
\begin{align}
 {\mathrm{P}}^{II}_{m,BS}  &=\mathrm{P}\left( \frac{\epsilon_{m'}}{ \rho|h_{m'}|^2}>1\right) +\mathrm{P}\left(  \frac{\epsilon_{m'}}{ \rho|h_{m'}|^2}<1,\frac{\rho|h_m|^2\left(1-\frac{\epsilon_{m'}}{ \rho|h_{m'}|^2}\right)}{\rho|h_{m'}|^2\frac{\epsilon_{m'}}{ \rho|h_{m'}|^2}+1}<\epsilon_m\right).
\end{align}
   Therefore, the outage probability can be expressed as follows:
 \begin{align}\label{case II 3}
 {\mathrm{P}}^{II}_{m,BS}
&= \mathrm{P}\left(x< \frac{L(d_{m'})\epsilon_{m'}}{ \rho }\right) +\mathrm{P}\left( \frac{L(d_{m'})\epsilon_{m'}}{ \rho }<x<\frac{ L(d_{m'})\epsilon_{m'}}{ \rho}  +\frac{\epsilon_m \left( \epsilon_{m'} +1\right) L(d_{m})}{\rho}\right).
\end{align}
 By applying the same steps  as in the proof of Lemma \ref{lemma 4} for finding   $\mathrm{P}(E_1)$ and $\mathrm{P}(E_3)$, the outage probability can be obtained as follows:
\begin{align}\label{case ii2}
 {\mathrm{P}}_{m',BS}^{II} = {\mathrm{P}}_{m,BS}^{II}  &=1 - \Upsilon_1\left(\frac{\epsilon_{m'}}{\rho}\right)\Upsilon_2
\left(\frac{\epsilon_m(1+\epsilon_{m'}) }{\rho }\right).
\end{align}

{\it Remark 3:} The two considered cases strike different tradeoffs between the outage performance of the two users.  Case I can ensure that the QoS at user $m'$ is strictly met, and  therefore  user $m'$ will experience  a lower outage probability in Case I, which can be confirmed by the fact that ${\mathrm{P}}_{m',BS}^{I} < {\mathrm{P}}_{m',BS}^{II}$,
due to $\Upsilon_2
\left(\frac{\epsilon_m(1+\epsilon_{m'}) }{\rho }\right)\leq 1$. On the other hand, Case II does not require that the message of  user $m'$ arrives at the base station with a stronger signal strength since the base station will decode  the message from user $m$ first. This is important to  avoid the problem of using  too much power   for  compensating  the huge path loss of the channel of  user $m'$. As a result,
 more power is allcoated to user $m$   compared to Case I, and hence, user $m$ experiences better outage performance in Case II, i.e., ${\mathrm{P}}_{m,BS}^{I} > {\mathrm{P}}_{m,BS}^{II}$. This can be shown by comparing  \eqref{for compare} with \eqref{case II 3} and by considering
 \begin{align}  L(d_{m})\epsilon_{m}   + \epsilon_{m'} \left( \epsilon_{m} +1\right) L(d_{m'}) <  L(d_{m'})\epsilon_{m'}   + \epsilon_m \left( \epsilon_{m'} +1\right) L(d_{m}) .\end{align}
\section{Numerical Studies}
In this section, the performance of the proposed NOMA framework is investigated by using computer simulations. The performance of three benchmark schemes, termed {\it MIMO-OMA without precoding}, {\it MIMO-OMA with precoding}, and {\it MIMO-NOMA without precoding}, is shown  in Fig. \ref{figure2}, in order to better illustrate the performance gain of the proposed framework. The design for the two schemes without precoding can be found in \cite{Zhiguo_mimoconoma}. The MIMO-OMA scheme with precoding serves $M$ users during  each orthogonal channel use, e.g., one time slot, whereas $2M$ users are served simultaneously by the proposed scheme. For MIMO-OMA with precoding, the design of the detection vectors was  obtained by following the algorithm proposed in Table~\ref{alg:stuff}, where the users will carry out antenna selection in each iteration. The framework proposed in this paper is termed {\it SA-MIMO-NOMA}. The path loss exponent is set as $\alpha=3$. The size of $\mathcal{D}_1$ and $\mathcal{D}_2$ is determined by  $r=20$m, and  $r_1=10$m. The parameter for the bounded path loss model is set as $r_0=1$.

\begin{figure}[!htp]\vspace{-1em}
\begin{center} \subfigure[ Outage Sum rate ]{\label{fig set comparison
b1}\includegraphics[width=0.42\textwidth]{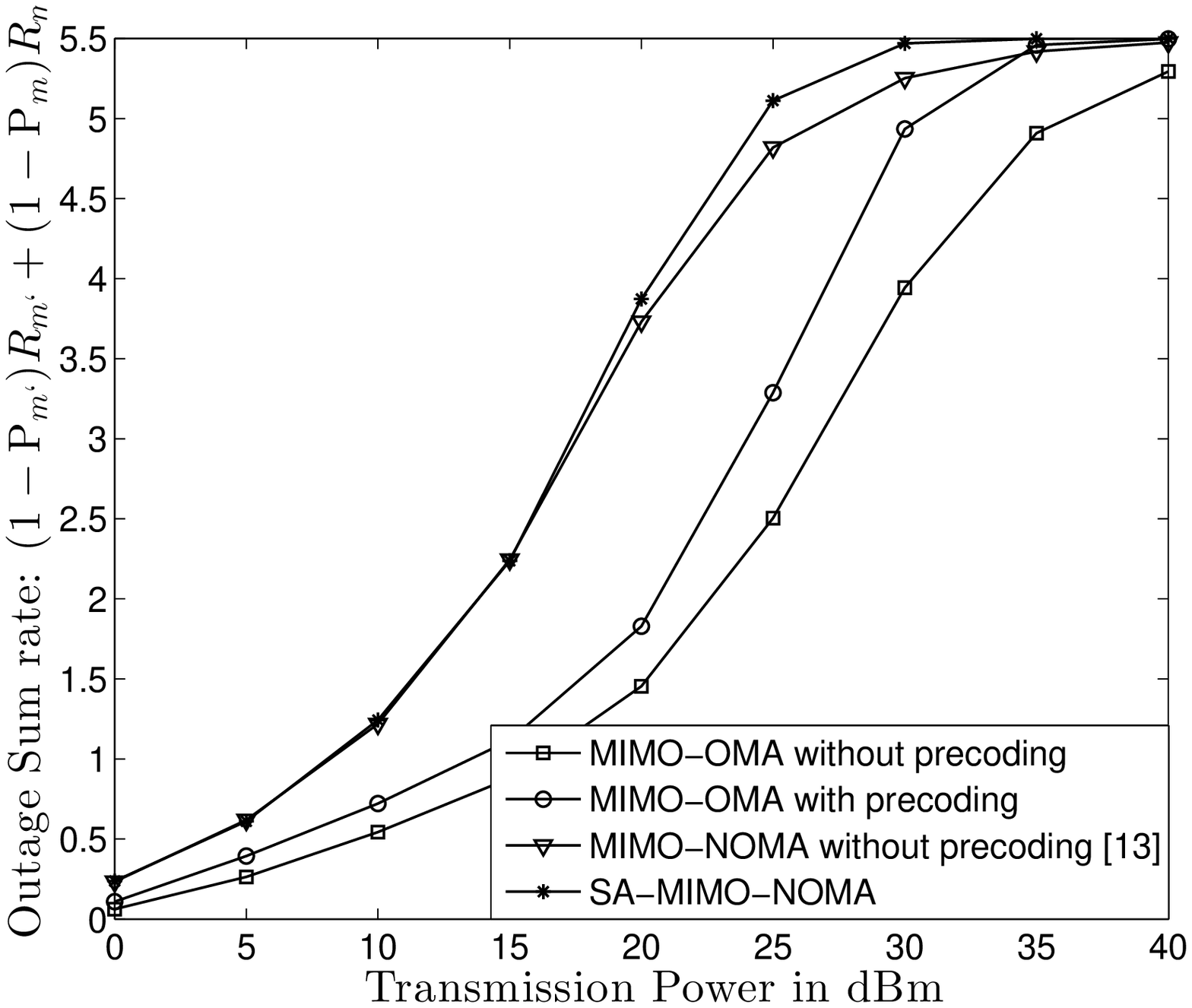}}
\subfigure[Outage Probabilities  ]{\label{fig set comparison
b2}\includegraphics[width=0.47\textwidth]{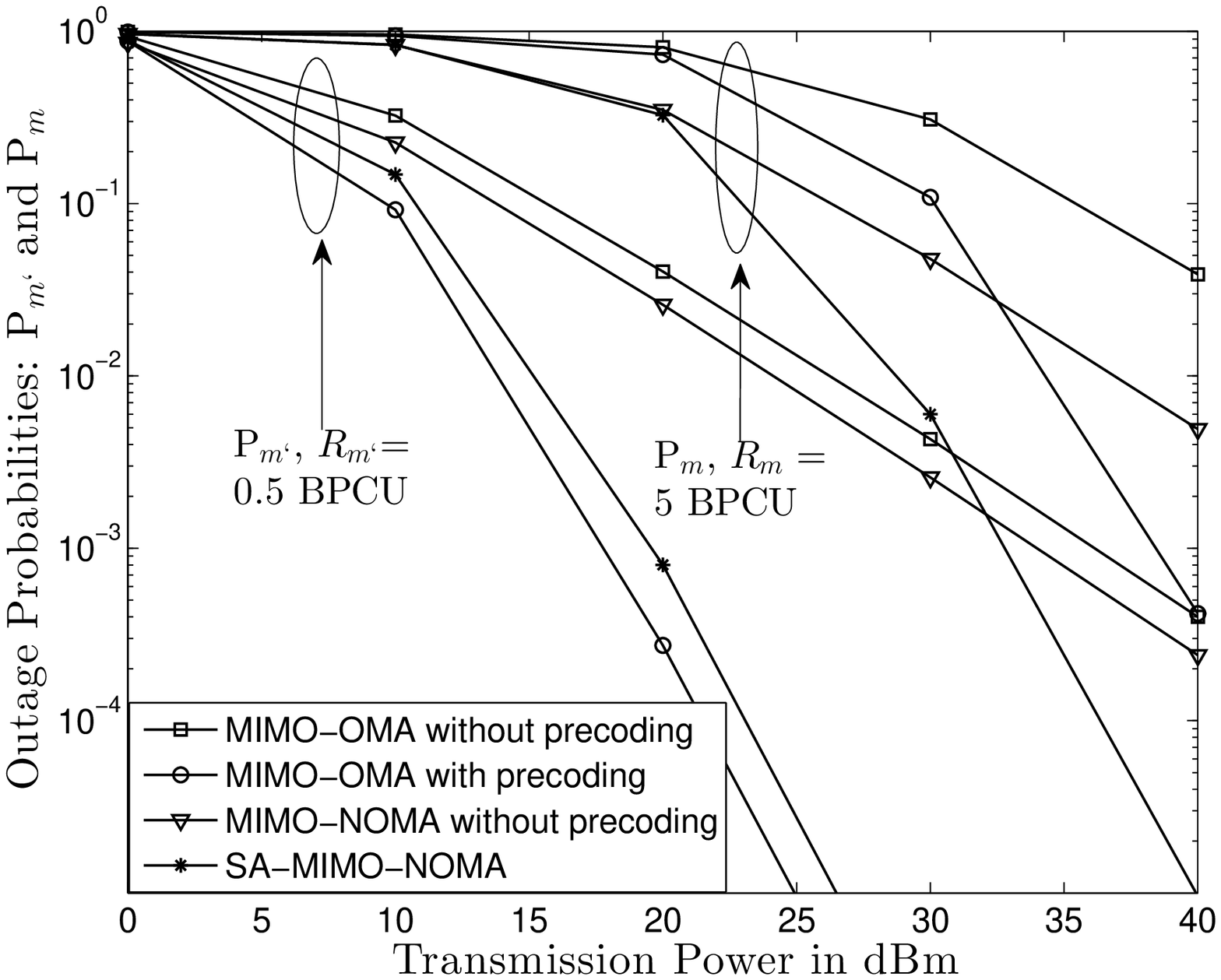}}\vspace{-1em}
\end{center}
  \caption{Performance comparison with  the three benchmark schemes for downlink transmission. $R_{m}=5$ BPCU and $R_{m'}=0.5$ BPCU.    $r=20$m and $r_1=10$m. $M=N=3$. $r_0=1$m. $a_{m'}=\frac{3}{4}$.  The path loss exponent is $\alpha=3$. The noise power is $-30$dBm and the interference power is $\rho_I=0$.  }\vspace{-1em} \label{figure2}
\end{figure}

Since the benchmark   schemes were proposed for the interference-free scenario,   Fig.~\ref{figure2} shows the performance comparison of the four schemes for   $\rho_I=0$.  In Fig. \ref{fig set comparison b1}, the downlink  outage sum rate, defined as $R_{m'}(1-\mathrm{P}_{m'})+R_{m}(1-\mathrm{P}_{m})$, is shown as a function of transmission power, and the corresponding outage probabilities are studied in Fig. \ref{fig set comparison b2}. As can be seen from the figures, the two NOMA schemes can achieve larger outage sum rates compared to the two OMA schemes, which demonstrates the superior spectral efficiency of NOMA. In Fig. \ref{fig set comparison b2}, the two schemes with precoding can achieve better outage performance than the two schemes without precoding, due to the efficient use of the  degrees of freedom at the base station. Comparing SA-MIMO-NOMA with the MIMO-NOMA scheme proposed in \cite{Zhiguo_mimoconoma}, one can observe that their outage sum rate performances are similar, but SA-MIMO-NOMA can offer  much better reception  reliability, particularly with  high transmission power. In terms of individual outage probability, SA-MIMO-NOMA can ensure   a lower  outage probability   at user $m$, i.e., a smaller  $\mathrm{P}_m$, compared to  the MIMO-OMA scheme with precoding,  but results in  performance degradation  for the outage probability at user $m'$, i.e., an increase of $\mathrm{P}_{m'}$. This is consistent with the finding in \cite{Zhiguo_CRconoma} which shows that the  NOMA user with poorer channel conditions will suffer some performance loss due to the co-channel interference from its partner.

\begin{figure}[!htp]\vspace{-1em}
\begin{center} \hspace{-3em}\subfigure[ Exact expressions ($R_{m'}=R_m=1$ BPCU)]{\label{fig set comparison
b3}  \includegraphics[width=0.48\textwidth]{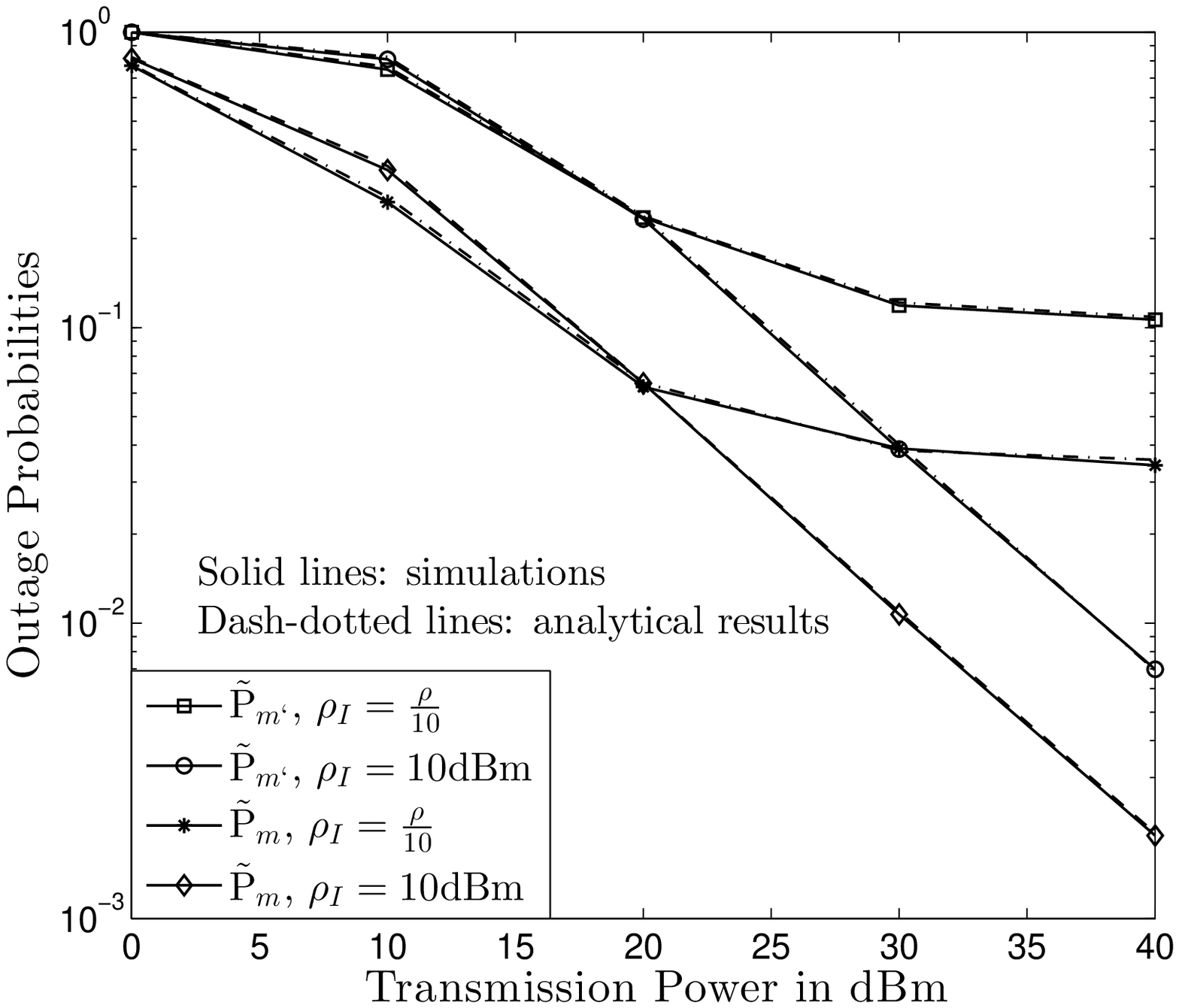}}
 \subfigure[Approximation ($\rho_I=-10$dBm)]{\label{fig set comparison
b4} \includegraphics[width=0.48\textwidth]{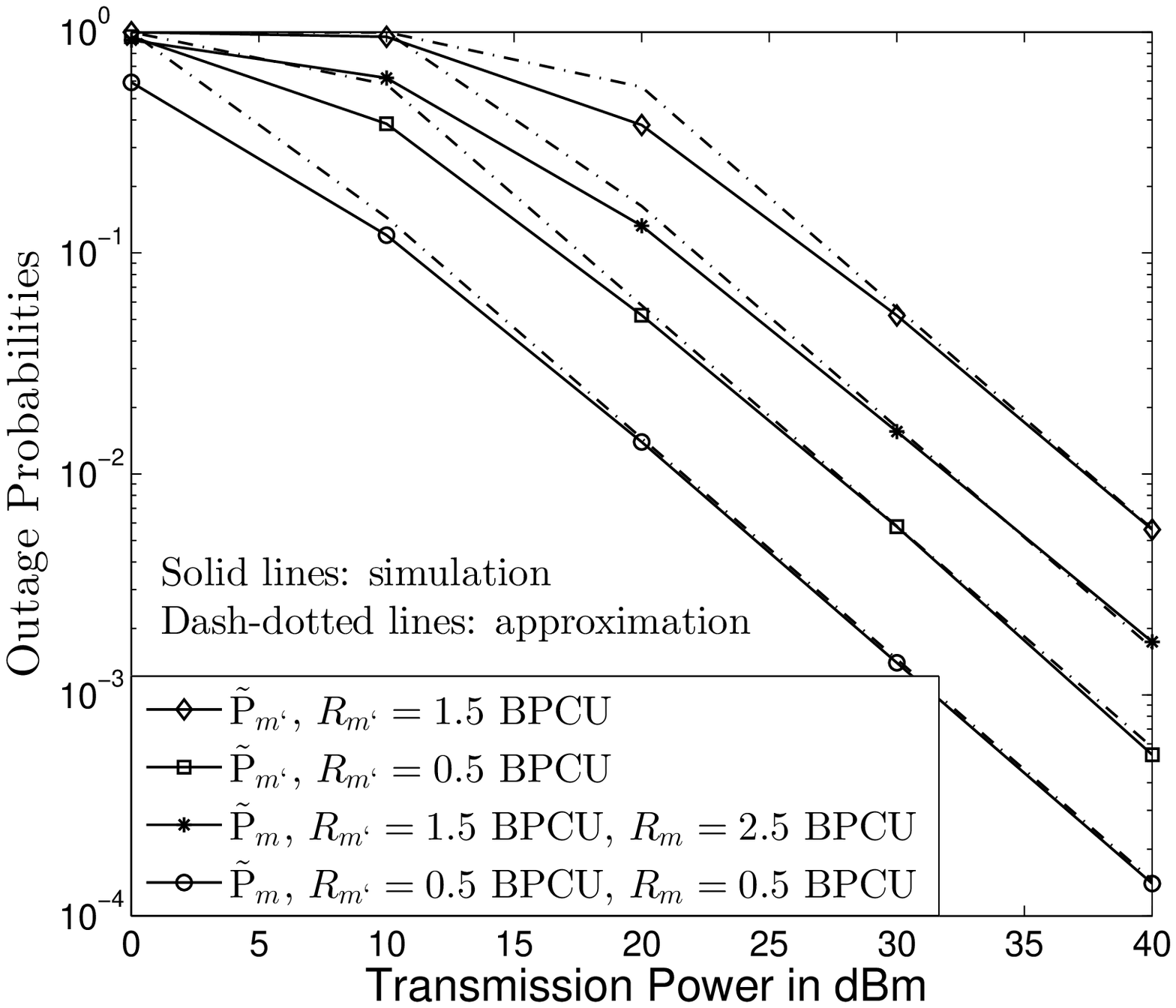}}\vspace{-1em}
\end{center}
  \caption{Outage probabilities $\tilde{\mathcal{P}}_{m'}$ and $\tilde{\mathcal{P}}_{m}$ for downlink transmission.   $\lambda_I= 10^{-4}$, $\delta=1$, $r=20$m,  $r_1=10$m, $M=N=2$,  $r_0=1$m, and  $a_{m'}=\frac{3}{4}$. The path loss exponent is $\alpha=3$ and the noise power is $-30$dBm.   The analytical results are based on Lemmas \ref{lemma1} and \ref{lemma2}.  }\label{fig3}\vspace{-1em}
\end{figure}

\begin{figure}[!htbp]\centering\vspace{-1em}
    \epsfig{file=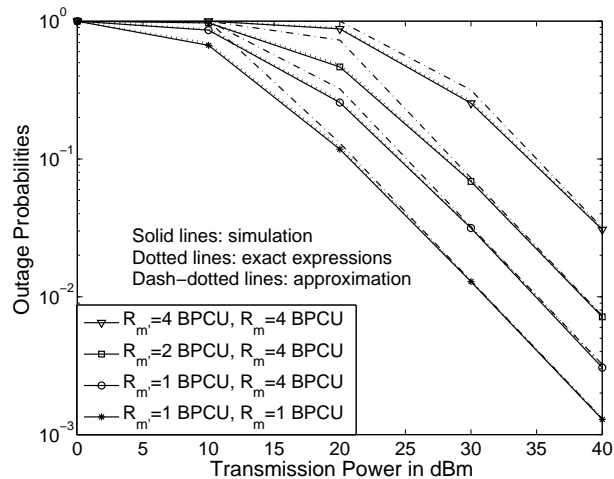, width=0.46\textwidth, clip=}\vspace{-1em}
\caption{ Outage probability  $\tilde{\mathrm{P}}_{m}$ for   cognitive radio downlink transmission.  $r=20$m,  $r_1=10$m,  $r_0=1$m, $\delta=1$, $\rho_I=0$, and  $M=N=2$. The noise power is $-30$dBm. The analytical results and the approximations are based on Lemma \ref{lemma 4}.    \vspace{-1em} }\label{fig5}
\end{figure}
In Fig. \ref{fig3}, the accuracy of the analytical results developed in Lemmas \ref{lemma1} and \ref{lemma2} for downlink transmission is verified. As can be seen from Fig. \ref{fig set comparison b3}, the exact expression developed in Lemma \ref{lemma1} perfectly matches the computer simulations, and the asymptotic results developed in Lemma \ref{lemma1} are also accurate at high SNR, as shown in Fig. \ref{fig set comparison b4}. The accuracy of Lemma \ref{lemma2} can be confirmed similarly. Note that error floors appear   when increasing $\rho_I$ in Fig. \ref{fig set comparison b3}, which is expected  due to the strong co-channel interference caused by the randomly deployed interferers. 

In Fig. \ref{fig5}, the performance of the cognitive radio power allocation scheme proposed in Section \ref{cognitive section} is studied.  In particular, given the target data rate at  user $m'$, the power allocation coefficients can be calculated opportunistically according to \eqref{choice 1 cg}. As can be seen from the figure, the probability for this NOMA system to support the secondary user, i.e., user $m$, with a target data rate of $R_{m}$ approaches one at high SNR. Note that with OMA, user $m$ cannot be admitted into the channel occupied by user $m'$, and with cognitive radio NOMA, one additional user, user $m$, can be served without degrading the outage performance of the primary user, i.e., user $m'$. 
 
 In Fig. \ref{figx11}, the impact of the number of   user antennas on the outage probability is studied. As can be seen from the figure, by increasing the number of the user antennas, the outage probability is decreased, since the dimension of the null  space, $\mathbf{U}_m$, defined in \eqref{dtection 1}, is increased and there are more possible choices for the detection vectors.  Furthermore, the slope of the outage curves is also increased, which indicates an increase of the achieved diversity order and hence confirms the findings of Lemma \ref{lemma 5}.

\begin{figure}[!htbp]\centering\vspace{-0.1em}
    \epsfig{file=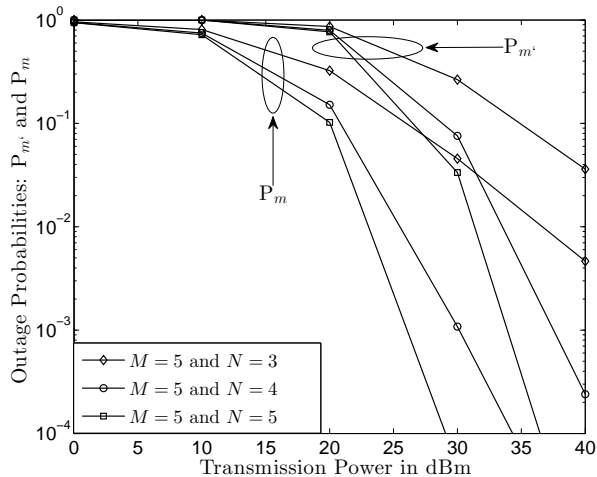, width=0.46\textwidth, clip=}\vspace{-1em}
\caption{ Impact of the number of the user antennas on the downlink outage probabilities ${\mathrm{P}}_{m}$ and $ {\mathrm{P}}_{m'}$. $R_m=4$ BPCU, $R_{m'}=1.9$ BPCU, $a_{m'}=\frac{3}{4}$, $\rho_I=0$,  $r=20$m, $r_1=10$m, and $r_0=1$m. The noise power is $-30$dBm. }\label{figx11}\vspace{-1em}
\end{figure}

\begin{figure}[!htbp]\centering\vspace{-1em}
    \epsfig{file=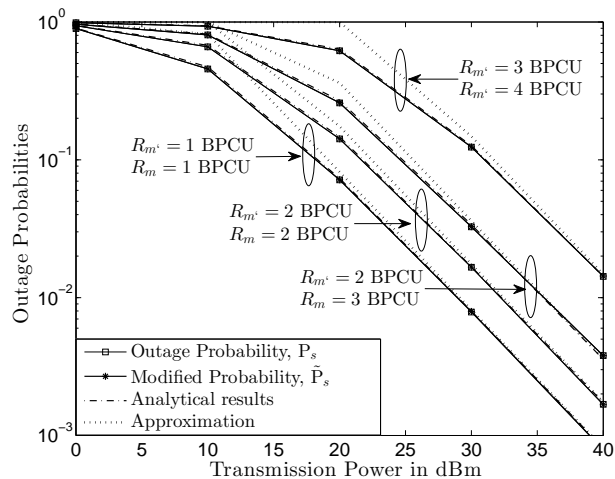, width=0.46\textwidth, clip=}\vspace{-1em}
\caption{ Outage probabilities for the uplink sum rate  ${\mathrm{P}}_{s}$ and $\tilde{\mathrm{P}}_{s}$. $r=20$m, $r_1=10$m,  $r_0=1$m, $\delta=1$, $\lambda_I= 10^{-4}$,   $\rho_I=10$dB, and  $M=N=2$. The noise power is $-30$dBm, and the interference power is $\rho_I=-10$dBm.  The analytical results and the approximations are based on \eqref{uplink1} and \eqref{uplink 2}, respectively.  }\label{fig6}\vspace{-1em}
\end{figure}
\begin{figure}[!htbp]\centering\vspace{-1em}
    \epsfig{file=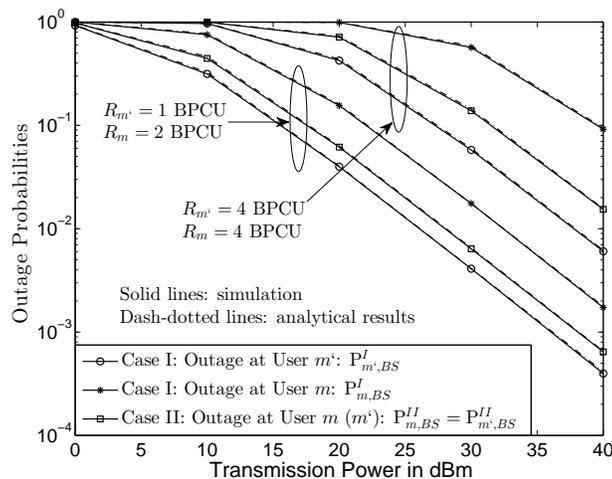, width=0.46\textwidth, clip=}\vspace{-1em}
\caption{ Uplink outage probability for user $m$ with the cognitive radio constraint. $r=4$m, $r_1=2$m, $r_0=1$m, $\rho_I=0$, and  the noise power is $-30$dBm.  The analytical results and the approximations are based on \eqref{case 12} and \eqref{case ii2}, respectively.  }\label{fig7}\vspace{-1em}
\end{figure}

The performance of the proposed NOMA framework for uplink transmission is demonstrated in Figs. \ref{fig6} and \ref{fig7}. In particular, in Fig. \ref{fig6}, the outage probability for the sum rate is investigated, and in Fig. \ref{fig7} the performance of the proposed cognitive radio uplink schemes is studied. As can be observed from both figures, the developed analytical results perfectly match  the computer simulation results, which demonstrates the accuracy of the developed analytical framework. It is worth pointing out that the modified probability $\tilde{\mathrm{P}}_s$ with $\delta=1$ provides an accurate approximation for $\mathrm{P}_s$.  An interesting observation from Fig. \ref{fig7} is that Cases I and  II offer different performance advantages.   In terms of $\mathrm{P}_{m'}$, Case I can offer a lower outage probability compared to Case II, however it results in a loss in   outage performance for  user $m$. In practice, if the QoS requirement at  user $m'$ is strict, Case I should be used, since the outage probability realized by Case I is exactly the same as when the entire bandwidth   is solely occupied by user $m'$. Otherwise, the use of Case II is more preferable  since the outage performance for  user $m$ can be improved and the system will not spend exceedingly high powers to compensate the user with poorer channel conditions. One can also observe that, for Case I with $R_{m'}=R_m$,  the outage performance for user $m$ is worse than that of  user $m'$, although  user $m$ is closer to the base station. The reason for this is because  in Case I, the power is allocated to  user $m'$ first, and user $m$ is served only if there is any power left. Therefore, the outage probability of  user $m$ will be at least the same as that of  user $m'$, as discussed in Section \ref{section uplink}.

\section{Conclusions}
In this paper, we have proposed a signal alignment based framework which is applicable to both  MIMO-NOMA downlink and uplink transmission. By applying tools from stochastic geometry, the impact of the random locations of the users and   interferers has been captured, and   closed-form expressions for the outage probability achieved by the proposed framework have been developed to facilitate performance evaluation. In addition to fixed power allocation, a more opportunistic power allocation strategy inspired by cognitive ratio networks has also been  investigated. Compared to the existing MIMO-NOMA work, the proposed framework is not only more general, i.e., applicable to both uplink and downlink transmissions, but also offers a significant performance gain in terms of reception reliability.  In this paper, it has been  assumed that  global  CSI is available, which may introduce a significant training overhead in practice. An important future direction is to study how  MIMO-NOMA transmission can be realized with limited CSI feedback.

\appendices
\section{Proof for Lemma \ref{lemma1}}
First, we  rewrite  the considered  probability $\tilde{\mathrm{P}}_{m'}$  as follows:
\begin{align}
\tilde{\mathrm{P}}_{m'}= \mathrm{P}\left( \frac{\frac{\rho \alpha_{m'}^2}{L(d_{m'})(\mathbf{G}^{-1}\mathbf{G}^{-H})_{m,m}}}{\frac{\rho
 \alpha_{m}^2}{L(d_{m'})(\mathbf{G}^{-1}\mathbf{G}^{-H})_{m,m}}+2+2\delta I_{m'}} < \epsilon_{m'}\right).
\end{align}

In order to calculate  $\tilde{\mathrm{P}}_{m'}$,   the density functions for the three parameters, $d_{m'}$, $I_{m'}$ and $\frac{1 }{(\mathbf{G}^{-1}\mathbf{G}^{-H})_{m,m}}$ have to be found.
Recall that the factor   $\frac{1 }{(\mathbf{G}^{-1}\mathbf{G}^{-H})_{m,m}}$ can be written as follows \cite{Rupp03}:
\begin{align}
\frac{1 }{(\mathbf{G}^{-1}\mathbf{G}^{-H})_{m,m}} = \mathbf{g}_m^H\left(\mathbf{I}_M - \Theta_m\right)\mathbf{g}_m,
\end{align}
where $\Theta_m=\tilde{\mathbf{G}}_m(\tilde{\mathbf{G}}_m^H
\tilde{\mathbf{G}}_m)^{-1}\tilde{\mathbf{G}}_m^H$ and $\tilde{\mathbf{G}}_m$ is obtained from $\mathbf{G}$ by removing its $m$-th row. If $\mathbf{g}_m$ is complex Gaussian distributed, the density function of  $\frac{1 }{(\mathbf{G}^{-1}\mathbf{G}^{-H})_{m,m}}$  will be exponentially distributed. This can be shown as follows. First, note that the projection matrix $(\mathbf{I}_M-\Theta_m)$ is an idempotent matrix and has eigenvalues which are either  zero or one.
Second, recall  that each row of $\mathbf{G}$ is generated from an $M\times 2N$ complex Gaussian matrix $\begin{bmatrix}\mathbf{G}_m^H & \mathbf{G}_{m'}^H  \end{bmatrix}$, i.e., \begin{align}
\mathbf{g}_m &= \frac{1}{2} \begin{bmatrix}\mathbf{G}_m^H & \mathbf{G}_{m'}^H  \end{bmatrix}\begin{bmatrix}\mathbf{v}_m^H & \mathbf{v}_{m'}^H  \end{bmatrix}^H=\frac{1}{2} \begin{bmatrix}\mathbf{G}_m^H & \mathbf{G}_{m'}^H  \end{bmatrix}\mathbf{U}_m \mathbf{x}_m.
\end{align}
Hence, provided that $\mathbf{x}_m$ is a randomly generated and normalized vector, the application of Proposition 1 in \cite{Dingtong11} yields the following
\begin{align}
\mathbf{g}_m \quad \sim  \quad {\rm CN}(0, \mathbf{I}_{M}),
\end{align}
i.e., $\mathbf{g}_m$ is still an $M\times 1$ complex Gaussian (CN) vector. Therefore, $\frac{1 }{(\mathbf{G}^{-1}\mathbf{G}^{-H})_{m,m}}$ is indeed exponentially distributed, and the outage probability can be expressed as follows:
\begin{align}
\tilde{\mathrm{P}}_{m'}= \mathcal{E}_{I_{m'},d_{m'}}\left\{1 - e^{- 2\phi_{m'}L(d_{m'})}\underset{Q_1}{\underbrace{e^{- 2\delta \phi_{m'}L(d_{m'}) I_{m'}}}}\right\},
\end{align}
which is conditioned on $ \alpha_{m'}^2 > \alpha_{m}^2 \epsilon_{m'}$. Otherwise, $\tilde{\mathrm{P}}_{m'}$ is always one.

Since the homogenous PPP $\Psi_I$ is stationary, the statistics of the interference    seen by user $m'$ is the same as that seen by any other receiver, according to Slivnyak's theorem \cite{5895051}. Therefore, $I_{m'}$ can be equivalently evaluated by  focusing on the interference reception seen at  a node located at the origin, denoted by  $
 I_{0}= \underset{j\in\Psi_I}{\sum}\frac{\rho_I}{L\left(d_{I_j}\right)}$, where $d_{I_j}$ denotes the distance between the origin and the $j$-th interference source. As a result, the expectation of $Q_1$ with respect to $I_{m'}$ can be expressed as follows: \cite{Haenggi}, \cite{6750425}
\begin{align}\label{eq 86}
\mathcal{E}_{I_{m'}}\left\{Q_1\right\}&= \mathcal{E}_{I_{m'}}\left\{ e^{- 2 \delta \phi_{m'}L\left(d_{m'}\right) \underset{j\in\Psi_I}{\sum}\frac{\rho_I}{L\left(d_{I_j}\right)}}\right\}= {\rm exp}\left(-\lambda_I \int_{t\in  {\mathcal{R}}^2} \left(1-e^{- 2 \delta \phi_{m'}\rho_IL\left(d_{m'}\right)L(p)}\right) dp \right),
\end{align}
where $p$ denotes the coordinate of the interference source, and  $d$ denotes the distance. Note that   distance $d$ is determined by the node location $p$. After changing to polar coordinates, the factor $\mathcal{E}_{I_{m'}}$ can be calculated as follows:
\begin{align}\nonumber
\mathcal{E}_{I_{m'}}\left\{Q_1\right\}&= {\rm exp}\left(-\pi\lambda_Ir_0^2 \left(1-e^{- \frac{\beta_{m'}(d_{m'})}{r_0^\alpha}}\right)  \right) {\rm exp}\left(-2\pi\lambda_I \int_{r_0}^\infty  \left(1-e^{- \frac{\beta_{m'}(d_{m'})}{x^\alpha}}\right) xdx \right)\\  \label{incomplete} &  =  {\rm exp}\left(-\pi \lambda_I \beta_{m'}^{\frac{2}{\alpha}}
\gamma\left(\frac{1}{\alpha},\frac{\beta_{m'}}{r_0^\alpha}\right)\right),
\end{align}
where     $\beta_{m'}(d_{m'})$ is denoted by  $\beta_{m'}$ for notational simplicity.
Therefore, the outage probability can be expressed as follows:
\begin{align}
\tilde{\mathrm{P}}_{m'}=1- \mathcal{E}_{d_{m'}}\left\{ e^{- 2\phi_{m'}(d_{m'}^\alpha)}\mathcal{E}_{I_{m'}}\left\{Q_1\right\}\right\}.
\end{align}

 Recall that  user $m'$ is uniformly distributed  in the ring $\mathcal{D}_2$. Therefore, the above expectation with respect to $d_{m'}$ can be calculated as follows:
 \begin{align}
\tilde{\mathrm{P}}_{m'}=1-   \int_{p\in\mathcal{D}_2} e^{- 2\phi_{m'}L(d_{m'})}\mathcal{E}_{I_{m'}}\left\{Q_1\right\}\frac{dp}{\pi r^2-\pi r_1^2},
\end{align}
where    distance $d_{m'}$ is determined by the user location $p$. Changing again to polar coordinates, this probability can be expressed as follows:
\begin{align}\label{pp1}
\tilde{\mathrm{P}}_{m'}&=1-  \frac{2}{  r^2-  r_1^2} \int_{r_1}^r e^{- 2\phi_{m'}x^\alpha}\mathcal{E}_{I_{m'}}\left\{Q_1\right\}xdx.
\end{align}
Hence, the first part of the lemma is proved.

In the case that $\rho$ approaches   infinity and  $\rho_I$ is fixed, it is easy to verify that  $\phi_{m'}$, as well as $\beta_{m'}$, go to zero. Hence, the  incomplete Gamma function in \eqref{incomplete} can be approximated as follows:
\begin{align}
\gamma\left(\frac{1}{\alpha},\frac{\beta_{m'}}{r_0^\alpha}\right) &= \sum^{\infty}_{n=0}\frac{(-1)^n \left(\frac{\beta_{m'}}{r_0^\alpha}\right)^{\frac{1}{\alpha}+n}}
{n!\left(\frac{1}{\alpha}+n\right)}
\approx \alpha \left(\frac{\beta_{m'}}{r_0^\alpha}\right)^{\frac{1}{\alpha}}.
\end{align}
Therefore, the factor $\mathcal{E}_{I_{m'}}$ can be approximated as follows:
\begin{align}
\mathcal{E}_{I_{m'}}\left\{Q_1\right\}  &  \approx {\rm exp}\left( -\pi \lambda_I \beta_{m'}^{\frac{2}{\alpha}}
 \alpha \left(\frac{\beta_{m'}}{r_0^\alpha}\right)^{\frac{1}{\alpha}}\right) \triangleq  e^{- d_{m'}^\alpha\theta_{m'}}
,
\end{align}
where $\theta_{m'}=2\pi \lambda_I\delta \phi_{m'}\rho_I\frac{\alpha}{r_0}$.
 Using this approximation the outage probability can be simplified at high SNR as follows:
\begin{align}
\tilde{\mathrm{P}}_{m'}&\approx1-  \frac{2}{  r^2-  r_1^2} \int_{r_1}^r e^{- 2\phi_{m'}x^\alpha}e^{-  x^{\alpha}\theta_{m'}}xdx \\ \nonumber
&\approx1-  \frac{2}{  r^2-  r_1^2} \int_{r_1}^r \left(1- (2\phi_{m'}+\theta_{m'})x^\alpha\right)xdx
=\frac{2(2\phi_{m'}+\theta_{m'})}{  r^2-  r_1^2}\frac{\left(r^{\alpha+2}-r_1^{\alpha+2}\right)}{\alpha+2}.
\end{align}

For the special cause without co-channel interfere, i.e., $\rho_I=0$, the   probability in \eqref{pp1} can be simplified as follows:
\begin{align}
\tilde{\mathrm{P}}_{m'}&=1-  \frac{2}{  r^2-  r_1^2} \int_{r_1}^r e^{- 2\phi_{m'}x^\alpha} xdx  = 1 - \frac{1}{  r^2-  r_1^2} \int^{r^{\alpha}}_{r_1^{\alpha}} e^{-2\phi_{m'}y}dy^{\frac{2}{\alpha}}\\ \nonumber & = 1 - \frac{e^{-2\phi_{m'}}}{  r^2-  r_1^2} \left(  e^{-r^{\alpha}} r^2 - e^{-r_1^{\alpha}} r_1^2\right)-\frac{(2\phi_{m'})^{-\frac{2}{\alpha}}}{  r^2-  r_1^2} \left(\gamma\left(\frac{2}{\alpha}+1, 2\phi_{m'}r^{\alpha}\right)-\gamma\left(\frac{2}{\alpha}+1, 2\phi_{m'}r_1^{\alpha}\right)\right) ,
\end{align}
and the lemma is proved.

\section{Proof for Lemma \ref{lemma2}}
When $ \alpha_{m'}^2 > \alpha_{m}^2 \epsilon_{m'}$, the outage probability  $\tilde{\mathrm{P}}_{m}$ can be written as follows:
\begin{align}
\tilde{\mathrm{P}}_{m}&=  \mathrm{P}\left(|h_m|^2<2\phi_{m'}(1+\delta I_m)\right) +\mathrm{P}\left( |h_m|^2<2\phi_m(1+\delta I_m), |h_m|^2>2\phi_{m'}(1+\delta I_m)\right)\\ \nonumber &= \mathrm{P}\left(|h_m|^2<2\max\{\phi_m,\phi_{m'}\}(1+\delta I_m)\right) ,
 \end{align}

The reason why $\tilde{\mathrm{P}}_{m}$ is an upper bound on $ {\mathrm{P}}^o_{m}$ for $\delta\geq N$ can be explained as follows. Recall that the original outage probability $\mathrm{P}^o_{m}$ can be expressed as $ {\mathrm{P}}^o_{m}= \mathrm{P}\left(|h_m|^2<\max\{\phi_m,\phi_{m'}\}(|\mathbf{v}_m|^2+|\mathbf{v}_m^H\mathbf{1}_N|^2I_{m})\right)$.
Since $|\mathbf{v}_m^H\mathbf{1}_N|^2\leq N |\mathbf{v}_m|^2$ and $|\mathbf{v}_m|^2\leq 2$, we have $\mathrm{P}^o_{m}\leq \tilde{\mathrm{P}}_{m}$ if $\delta\geq N$. It is worth pointing out that a choice of $\delta=1$ is sufficient to yield a tight approximation  on  $\mathrm{P}^o_{m}$, as shown in Fig. \ref{bound}.

 Recall that $h_{m}= \frac{1 }{\sqrt{L(d_{m})(\mathbf{G}^{-1}\mathbf{G}^{-H})_{m,m}}} $. Comparing  $h_m$ to $h_{m'}$, we   find that the only difference between the two is the distance $d_m$ which is less than $r_1$. In addition, the statistics  of $I_m$ can be studied by using $I_0$ as explained in the proof of Lemma \ref{lemma1}.  Therefore, following   steps similar to those in the proof of Lemma \ref{lemma1},   the outage probability can be expressed as follows:
\begin{align}
\tilde{\mathrm{P}}_{m}= \mathcal{E}_{I_{m},d_{m}}\left\{1 - e^{- 2\tilde{\phi}_{m}L(d_{m})}\underset{Q_2}{\underbrace{e^{- 2\tilde{\phi}_{m}L(d_{m}) I_{m}}}}\right\}.
\end{align}
 It is straightforward to show that  the expectation of $Q_2$ can be obtained in the same way as that of $Q_1$, by replacing $\phi_{m'}$ with $\tilde{\phi}_m$. In addition, recall that  user $m$ is uniformly distributed  in the disc $\mathcal{D}_1$. Therefore, the outage probability  can be calculated as follows:
 \begin{align}
\tilde{\mathrm{P}}_{m}=1-   \int_{p\in\mathcal{D}_1} e^{- 2\tilde{\phi}_{m}L(d_{m})}\varphi_I(L(d_m))\frac{dp}{\pi  r_1^2},
\end{align}
where    distance $d_{m}$  is again determined by the user location $p$. Resorting to polar coordinates, the outage probability can be expressed as follows:
\begin{align}
\tilde{\mathrm{P}}_{m}&=1-  \frac{2}{   r_1^2} \int_{0}^{r_0} e^{- 2\tilde{\phi}_{m}r_0^\alpha}\varphi_I(L(x)xdx - \frac{2}{   r_1^2} \int_{r_0}^{r_1} e^{- 2\tilde{\phi}_{m}x^\alpha}\varphi_I(L(x))xdx.
\end{align}
If $\rho$ approaches  infinity and  $\rho_I$ is fixed, both $\beta_{m}$ and  $\tilde{\phi}_{m}$ go to zero. With this approximation, the  incomplete Gamma function in \eqref{incomplete} can be approximated as $
\mathcal{E}_{I_{m}}\left\{Q_1\right\}\approx  e^{- d_{m}^\alpha {\theta}_{m}}$,
where $\theta_{m}=2\pi \lambda_I\tilde{\phi}_{m}\rho_I\frac{\alpha}{r_0}$.
Hence, the outage  probability can be simplified at high SNR as follows:
\begin{align}
\tilde{\mathrm{P}}_{m}&\approx 1-  \frac{2}{   r_1^2} \int_{0}^{r_0} e^{- 2\tilde{\phi}_{m}r_0^\alpha} e^{- r_0^\alpha\theta_{m}}xdx- \frac{2}{   r_1^2} \int_{r_0}^{r_1} e^{- 2\tilde{\phi}_{m}x^\alpha} e^{- x^\alpha\theta_{m}}xdx
\\ \nonumber
&\approx 1- \frac{r_0^2}{r_1^2}\left( 1- 2\tilde{\phi}_{m}r_0^\alpha- r_0^\alpha\theta_{m}\right) - \frac{2}{   r_1^2} \int_{r_0}^{r_1} \left(1- (2\tilde{\phi}_{m}+\theta_m)x^\alpha\right)xdx \\\nonumber &\approx \frac{(2\tilde{\phi}_m+\theta_m)}{r_1^2(\alpha+2) }\left(\alpha r_0^{\alpha+2}+2r_1^{\alpha+2}\right),
\end{align}  and the lemma is proved.

\section{Proof for Lemma \ref{lemma 4}}

There are three types of outage events at  user $m$, as illustrated in the following:
\begin{itemize}
\item  $\bar{\alpha}_m^2=0$, i.e., all the power is consumed by  user $m'$ and no power is allocated to  user $m$. This event is denoted by $E_1$.
    \item  When $\bar{\alpha}_m^2>0$,  user $m$ cannot decode the message to  user $m'$. This event is denoted by $E_2$.
    \item When $\bar{\alpha}_m^2>0$,  user $m$ can decode the message to  user $m'$, but fails to decode its own message. This event is denoted by $E_3$.
\end{itemize}
The probability of $E_1$ can be expressed as follows:
\begin{align}
\mathrm{P}(E_1) = \mathrm{P}\left(\rho |h_{m'}|^2-2\epsilon_{m'}<0\right).
\end{align}
This probability can be straightforwardly obtained from the proof of Lemma \ref{lemma1} by replacing $\phi_{m'}$ with $\breve{\phi}_{m'}\triangleq \frac{2\epsilon_{m'}}{\rho}$. Therefore, $\mathrm{P}(E_1)$ can be expressed as follows:
\begin{align}\label{eq1}
&{\mathrm{P}}(E_1)   = 1 - \Upsilon_1\left(\frac{2\epsilon_{m'}}{\rho}\right).
\end{align}

When $\bar{\alpha}_m>0$,  $\mathrm{P}(E_2)=0$, since
\begin{align}\label{eq3}
\mathrm{P}(E_2)=&\mathrm{P}\left(\frac{\rho|h_{m}|^2(1-\bar{\alpha}_{m}^2)}{\rho|h_{m}|^2
 \bar{\alpha}_{m}^2+2}<\epsilon_{m'}\right)
 =\mathrm{P}\left(\rho|h_{m}|^2(1-\bar{\alpha}_{m}^2)<\epsilon_{m'}(\rho|h_{m}|^2
 \bar{\alpha}_{m}^2+2)\right)\\ \nonumber
 =&\mathrm{P}\left(\rho|h_{m}|^2<\bar{\alpha}_{m}^2\rho|h_{m}|^2(1+\epsilon_{m'})
 +2\epsilon_{m'}\right)
 \\ \nonumber
 =&\mathrm{P}\left(\rho|h_{m}|^2 |h_{m'}|^2<\rho |h_{m'}|^2 |h_{m}|^2-2 |h_{m}|^2\epsilon_{m'}
 +2 |h_{m'}|^2\epsilon_{m'} \right) =\mathrm{P}\left( |h_{m}|^2  <
  |h_{m'}|^2  \right)=0.
\end{align}

The probability for   event $E_3$ can be calculated as follows:
\begin{align}
\mathrm{P}(E_3) &= \mathrm{P}\left(\log\left(1+\frac{\rho}{2} |h_m|^2\alpha_m^2\right)<R_m, \bar{\alpha}_m>0\right)
= \mathrm{P}\left( |h_m|^2\frac{\rho |h_{m'}|^2-2\epsilon_{m'}}{2(1+\epsilon_{m'})  |h_{m'}|^2}<\epsilon_m, |h_{m'}|^2>\frac{2\epsilon_{m'}}{\rho}\right).
\end{align}
An important observation is that both   channel gains $h_m$ and $h_{m'}$ share the same small scale fading. Defining  $x=\frac{1}{(\mathbf{G}^{-1}\mathbf{G}^{-H})_{m,m}}$,   the outage probability can be expressed as follows:
\begin{align}
\mathrm{P}(E_3)  &= \mathrm{P}\left( \frac{x}{L(d_{m})}\frac{\rho \frac{x}{L(d_{m'})}-2\epsilon_{m'}}{(1+\epsilon_{m'})  \frac{x}{L(d_{m'})}}<2\epsilon_m, \frac{x}{L(d_{m'})}>\frac{2\epsilon_{m'}}{\rho}\right)\\ \nonumber  &= \mathrm{P}\left(\frac{2\epsilon_{m'}L(d_{m'})}{\rho}< x <\frac{2\epsilon_{m'}L(d_{m'})}{\rho} +\frac{2\epsilon_m(1+\epsilon_{m'}) L(d_{m})}{\rho} \right).
\end{align}

The above probability can be calculated as follows
\begin{align}
\mathrm{P}(E_3)  &=  \underset{p_{m'}\in \mathcal{D}_2}{ \int}   e^{-\frac{2\epsilon_{m'}L(d_{m'})}{\rho}}dp_{m'} - \hspace{-0.5em} \underset{p_m\in \mathcal{D}_1,p_{m'}\in \mathcal{D}_2}{\int\int} \hspace{-1.5em} e^{-\frac{2\epsilon_{m'}L(d_{m'})}{\rho}-\frac{2\epsilon_m(1+\epsilon_{m'}) L(d_{m})}{\rho}}dp_mdp_{m'},
\end{align}
where $p_m$ denotes the location of  user $m$. Since the users are uniformly distributed, the above probability can be expressed  as follows:
\begin{align}\label{eq2}
\mathrm{P}(E_3)  &=\frac{2}{(r^2-r_1^2)} \int_{r_1}^{r}e^{-\frac{2\epsilon_{m'} }{\rho y^{\alpha}}}ydy-\frac{4}{r_1^2(r^2-r_1^2)} \int_{0}^{r_1} e^{-\frac{2\epsilon_m(1+\epsilon_{m'}) }{\rho y^{\alpha}}}ydy \int_{r_1}^{r}e^{-\frac{2\epsilon_{m'}L(x)}{\rho}} xdx
\\ \nonumber  &=\Upsilon_1
\left(\frac{2\epsilon_{m'} }{\rho }\right) - \Upsilon_1\left(\frac{2\epsilon_{m'}}{\rho}\right)\Upsilon_2
\left(\frac{2\epsilon_m(1+\epsilon_{m'}) }{\rho }\right).
\end{align}
Combining \eqref{eq1}, \eqref{eq3}, and \eqref{eq2}, the first part of the lemma can be proved.
To obtain the high SNR approximation,  we have
\begin{align}
\Upsilon_1(y) & \approx 1+ \frac{1}{  r^2-  r_1^2} \left(  yr_1^{\alpha+2} -yr^{\alpha+2} \right) +\frac{y^{-\frac{2}{\alpha}}}{ (\frac{2}{\alpha}+1)(r^2-  r_1^2)}\left(\left( yr^{\alpha}\right)^{\frac{2}{\alpha}+1}-\left(yr_1^{\alpha}\right)^{\frac{2}{\alpha}+1}\right) \\ \nonumber & = 1-\frac{2y}{ (2+\alpha)( r^2-  r_1^2)} \left(  r^{\alpha+2} -r_1^{\alpha+2} \right) ,
\end{align}
when $y$  approaches zero,
and
\begin{align}\nonumber
\Upsilon_2(z) & \approx 1 - \frac{r_0^{2+\alpha} z}{r_1^2}- \frac{1}{    r_1^2}  \left(  zr_1^{\alpha+2} -  zr_0^{\alpha+2}\right) +\frac{z^{-\frac{2}{\alpha}}}{  (\frac{2}{\alpha}+1)  r_1^2}  \left(\left( zr_1^{\alpha}\right)^{\frac{2}{\alpha}+1}-\left( zr_0^{\alpha}\right)^{\frac{2}{\alpha}+1}\right) \\
& = 1 - \frac{r_0^{2+\alpha} z}{r_1^2}  -\frac{2z}{  (2+\alpha)  r_1^2} \left(  r_1^{\alpha+2} -  r_0^{\alpha+2}\right),
\end{align}
when $z$  approaches zero. By substituting the above approximations into \eqref{eq lemma 3},  the lemma is proved.
\section{Proof for Lemma \ref{lemma 5}}
We focus on the outage performance of  user $m'$ first.  
Given  the detection vector $\mathbf{v}_{m,i^*}$ chosen from   Table \ref{alg:stuff},  the outage probability   can be   upper bounded  as follows:
\begin{align}\nonumber
 {\mathrm{P}}_{m',i^*}&\leq  \mathrm{P}\left( \frac{\frac{\rho \alpha_{m'}^2}{L(d_{m'})(\bar{\mathbf{G}}_{i^*}^{-1}\bar{\mathbf{G}}_{i^*}^{-H})_{m,m}}}{\frac{\rho
 \alpha_{m}^2}{L(d_{m'})(\bar{\mathbf{G}}_{i^*}^{-1}\bar{\mathbf{G}}_{i^*}^{-H})_{m,m}}+2+2\delta I_{m'}} < \epsilon_{m'}\right)\\   &=\mathrm{P}\left( \gamma_{m,{i^*}}<2\phi_{m'}L(d_{m'})(1+\delta I_{m'})\right)
 \leq\mathrm{P}\left( \gamma_{{\rm min},{i^*}}<2\phi_{m'}L(d_{m'})(1+\delta I_{m'})\right).
\end{align}
According to the algorithm proposed in Table \ref{alg:stuff},
\begin{align}
\gamma_{{\rm min},i^*}= \max  \{\gamma_{{\rm min},1},\cdots,\gamma_{{\rm min},2N-M}\}.
\end{align}
Therefore, the outage probability   can be bounded as follows:
\begin{align}\nonumber
 {\mathrm{P}}_{m',i^*}  &\leq\left(\mathrm{P}\left( \gamma_{{\rm min},i}<2\phi_{m'}L(d_{m'})(1+\delta I_{m'})\right)\right)^{2N-M},
\end{align}
where the inequality follows from the fact that $\gamma_{{\rm min},i}$ and $\gamma_{{\rm min},j}$ are independent, since $\mathbf{g}_{m,i}$ and $\mathbf{g}_{m,j}$ are independent (Proposition 1 in \cite{Dingtong11}). The above outage probability can be further bounded as follows:
\begin{align} \label{xcc}
 {\mathrm{P}}_{m',i^*}  &\leq\left(M\mathrm{P}\left( \gamma_{m,i}<2\phi_{m'}L(d_{m'})(1+\delta I_{m'})\right)\right)^{2N-M}.
\end{align}

Following the same steps  as in the proof of Lemma \ref{lemma1}, the upper bound on the outage probability can be calculated as follows:
\begin{align}
{\mathrm{P}}_{m',i^*} &\leq M^{2N-M}\mathcal{E}_{I_{m'},d_{m'}}\left\{\left(1 - e^{- 2\phi_{m'}L(d_{m'})} e^{- 2\delta \phi_{m'}L(d_{m'}) I_{m'}}\right)^{2N-M}\right\}\\ \nonumber&\leq M^{2N-M} \sum^{2N-M}_{i=0}{2N-M \choose i}(-1)^i  \mathcal{E}_{I_{m'},d_{m'}}\left\{e^{- 2i\phi_{m'}L(d_{m'})} e^{- 2i\delta \phi_{m'}L(d_{m'}) I_{m'}} \right\},
\end{align}
which is conditioned on $ \alpha_{m'}^2 > \alpha_{m}^2 \epsilon_{m'}$.

After the expectation with respect to $I_{m'}$, the outage probability can be bounded as follows:
\begin{align}
{\mathrm{P}}_{m',i^*}&\leq M^{2N-M} \sum^{2N-M}_{i=0}{2N-M \choose i}(-1)^i    \mathcal{E}_{ d_{m'}}\left\{e^{- 2i\phi_{m'}L(d_{m'})}e^{-\pi \lambda_I (i\beta_{m'})^{\frac{2}{\alpha}}
\gamma\left(\frac{1}{\alpha},\frac{i\beta_{m'}}{r_0^\alpha}\right)} \right\}.
\end{align}

For  the case  of $\rho$ approaching   infinity and a fixed $\rho_I$, the upper bound on the outage probability can be approximated as follows:
\begin{align}
{\mathrm{P}}_{m',i^*}&\leq M^{2N-M} \sum^{2N-M}_{i=0}{2N-M \choose i}(-1)^i   \mathcal{E}_{ d_{m'}}\left\{e^{- 2i\phi_{m'}L(d_{m'})}e^{-\pi \lambda_I (i\beta_{m'})^{\frac{2}{\alpha}}\alpha
\left(\frac{i\beta_{m'}}{r_0^\alpha}\right)^{\frac{1}{\alpha}}} \right\}\\\nonumber
&\leq M^{2N-M} \sum^{2N-M}_{i=0}{2N-M \choose i}(-1)^i    \mathcal{E}_{ d_{m'}}\left\{ e^{-(i\theta_{m'}+2i\phi_{m'})d_m^\alpha} \right\}.
\end{align}
Using    polar coordinates, the upper bound can be calculated as follows:
\begin{align}
{\mathrm{P}}_{m',i^*}
&\leq M^{2N-M} \sum^{2N-M}_{i=0}{2N-M \choose i}(-1)^i  \frac{2}{  r^2-  r_1^2}\sum^{\infty}_{j=0} \int_{r_1}^r   \frac{(-1)^j(i\theta_{m'}+2i\phi_{m'})^jx^{j\alpha}}{j!}  xdx
\\
&= \frac{2 M^{2N-M}}{  r^2-  r_1^2} \sum^{2N-M}_{i=0}{2N-M \choose i}(-1)^i   \sum^{\infty}_{j=0}    \frac{(-1)^j(i\theta_{m'}+2i\phi_{m'})^j}{j!} \frac{\left(r^{j\alpha+2}-r_1^{j\alpha+2}\right)}{j\alpha+2} .
\end{align}

By exchanging the two sums in the above equation, the upper bound can be rewritten as follows:
\begin{align}
{\mathrm{P}}_{m',i^*}
&\leq   \frac{2 M^{2N-M}}{  r^2-  r_1^2} \sum^{\infty}_{j=0}    \frac{(-1)^j(\theta_{m'}+2\phi_{m'})^j}{j!}   \frac{\left(r^{j\alpha+2}-r_1^{j\alpha+2}\right)}{j\alpha+2} \sum^{2N-M}_{i=0}{2N-M \choose i}(-1)^i i^j
\\
& =  \frac{2 M^{2N-M}}{  r^2-  r_1^2} \sum^{\infty}_{j=2N-M}   \frac{(-1)^j(\theta_{m'}+2\phi_{m'})^j}{j!}  \frac{\left(r^{j\alpha+2}-r_1^{j\alpha+2}\right)}{j\alpha+2} \sum^{2N-M}_{i=0}{2N-M \choose i}(-1)^i i^j,
\end{align}
where the last step follows from the following fact $$ \sum^{2N-M}_{i=0}{2N-M \choose i}(-1)^i i^j=0$$
for $0\leq j \leq (2N-M-1)$ \cite{GRADSHTEYN}. Furthermore, note that both $\phi_{m'}$ and $\theta_{m'}$ approach   zero for the considered scenario, and $\sum^{2N-M}_{i=0}{2N-M \choose i}(-1)^i i^{2N-M}=(-1)^{2N-M}(2N-M)!$. Therefore, the upper bound on the outage probability can be approximated  as follows:
\begin{align}
&{\mathrm{P}}_{m',i^*}
\leq     \frac{2 M^{2N-M}}{  r^2-  r_1^2}   \frac{(-1)^{2N-M}(\theta_{m'}+2\phi_{m'})^{2N-M}}{(2N-M)!}  \frac{\left(r^{(2N-M)\alpha+2}-r_1^{(2N-M)\alpha+2}\right)}{(2N-M)\alpha+2}(-1)^{2N-M}(2N-M)!  \\ \nonumber&= \frac{2[M(\theta_{m'}+2\phi_{m'})]^{2N-M}\left(r^{(2N-M)
\alpha+2}-r_1^{(2N-M)\alpha+2}\right)}{( r^2-  r_1^2)((2N-M)\alpha+2)}  \sim \frac{1}{\rho^{2N-M}}.
\end{align}
The result for  user $m$ can be proved using    steps similar to the ones above.

The result for  a random detection vector can be obtained by replacing  $(2N-M)$ with $1$ in the above expression, and the corresponding upper bound becomes
\begin{align}
&{\mathrm{P}}_{m',i^*}
\leq    \frac{2M[\theta_{m'}+2\phi_{m'}]\left(r^{
\alpha+2}-r_1^{\alpha+2}\right)}{( r^2-  r_1^2)(\alpha+2)} .
\end{align}
which is exactly the same result as the one shown in Lemma \ref{lemma1}, except for the extra term $M$ which was introduced by   upper bounding  the outage probability in \eqref{xcc}. Hence, the proof is completed.
\linespread{1}
 \bibliographystyle{IEEEtran}
\bibliography{IEEEfull,trasfer}
  \end{document}